\documentclass[12pt,reqno]{elsarticle}
\usepackage{geometry} 
\usepackage{natbib}
\usepackage{mathrsfs} 
\usepackage[latin1]{inputenc}
\usepackage{setspace}

\usepackage{amsmath} 
\usepackage{amsthm}
\usepackage{amssymb} 
\usepackage{calc}
\usepackage{graphicx}
\usepackage{enumerate}
\usepackage{latexsym}
\usepackage{version}
\usepackage{lineno,hyperref}
\usepackage{pbox}
\usepackage{floatrow}
\usepackage{bibentry}
\usepackage[labelformat = empty,position=top]{subcaption}
\usepackage[font=small]{caption}
\usepackage[export]{adjustbox}

\usepackage{tikz}
\usetikzlibrary{shapes,arrows,chains,automata}
\usetikzlibrary{positioning}

\newcommand{\tab}{\hspace*{0.4in}}

\newcommand{\bi}{\begin{itemize} \vspace*{-0.1in}}
	\newcommand{\ei}{\end{itemize}}

\newcommand{\be}{\begin{equation}}
\newcommand{\ee}{\end{equation}} 
\newcommand{\ba}{\begin{equation} \begin{aligned}}
\newcommand{\ea}{\end{aligned} \end{equation}}

\newtheorem{thm}{Theorem}

\newtheorem{prop}{Proposition}

\newcommand{\beginsupplement}{%
	\setcounter{table}{0}
	\renewcommand{\thetable}{S\arabic{table}}%
	\setcounter{figure}{0}
	\renewcommand{\thefigure}{S\arabic{figure}}%
}

\begin{document}
	
	\begin{frontmatter}
		
		This is a preprint of\\\\\\
		Holden, M.H. and E McDonald-Madden. 2017. High prices for rare species can drive large populations extinct: the anthropogenic Allee effect revisited. \textbf{\textit{Journal of Theoretical Biology}. 429:170-180.}\\\\\\
		Please use the above citation, and see the post-production version of the paper at \url{http://www.sciencedirect.com/science/article/pii/S0022519317302916}\\\\\\ 
		
		\title{High prices for rare species can drive large populations extinct: the anthropogenic Allee effect revisited}

		\author{Matthew H. Holden$^{1,2,*}$,\\Eve McDonald-Madden$^{1,3}$}
		\address{
			1.Australian Research Council Centre of Excellence for Environmental Decisions, University of Queensland, St. Lucia, QLD 4072, Australia\\
			2. School of Biological Sciences, University of Queensland, St. Lucia, QLD 4072, Australia\\
			3. School of Earth \& Environmental Science, University of Queensland, St. Lucia, Queensland, 4072, Australia\\
			*Corresponding author: m.holden1@uq.edu.au}
		
		\begin{abstract}
			Consumer demand for plant and animal products threatens many populations with extinction. The anthropogenic Allee effect (AAE) proposes that such extinctions can be caused by prices for wildlife products increasing with species rarity. This price-rarity relationship creates financial incentives to extract the last remaining individuals of a population, despite higher search and harvest costs. The AAE has become a standard approach for conceptualizing the threat of economic markets on endangered species. Despite its potential importance for conservation, AAE theory is based on a simple graphical model with limited analysis of possible population trajectories. By specifying a general class of functions for price-rarity relationships, we show that the classic theory can understate the risk of species extinction. AAE theory proposes that only populations below a critical Allee threshold will go extinct due to increasing price-rarity relationships. Our analysis shows that this threshold can be much higher than the original theory suggests, depending on initial harvest effort. More alarmingly, even species with population sizes above this Allee threshold, for which AAE predicts persistence, can be destined to extinction. Introducing even a minimum price for harvested individuals, close to zero, can cause large populations to cross the classic anthropogenic Allee threshold on a trajectory towards extinction. These results suggest that traditional AAE theory may give a false sense of security when managing large harvested populations.
		\end{abstract}
		
		
	\end{frontmatter}
	
	\doublespacing
	
	\section{Introduction}
	Overexploitation is one of the greatest threats to the conservation of endangered species \cite{Maxwell2016}. A common explanation for drastic population declines due to over-harvest is the \textit{tragedy of the commons} \cite{Hardin1968}. Under open-access conditions, where profits go to individual users, and losses in future harvest are shared by all users, each user has an incentive to drive the resource to low levels \cite{Colin2010}. However, in classic open-access harvest models, Tragedy of the Commons does not explain population extinctions, because the price-per-unit harvest is fixed and therefore, the cost of finding and extracting rare individuals eventually exceeds the price achieved from selling the harvested resource \citep{Clark1990, Colin2010}. 
	
	\tab Prices are not fixed; they depend on market dynamics that reflect how much consumers are willing to pay for the resource. For food, art, and other consumer goods, derived from harvested plants and animals, prices often increase with species rarity \cite{Courchamp2006a,angulo2009,gault2008,purcell2014,hinsley2015}. Products from rare species become luxury goods, status symbols, or financial investments for the wealthy \citep[e.g. rhino horn,][]{gao2016}, where exorbitant prices paid by consumers incentivize harvesters to absorb the high costs of searching for and killing the last few individuals. This phenomenon is called the ``anthropogenic Allee effect" (AAE) where species rarity increases price, and therefore the incentive to harvest, driving small populations extinct. It is named after the classic ecological concept of a \textit{strong Allee effect}, where populations decline towards zero only if they start below a critical threshold size [due to mate limitation for example] \cite{Stephens1999,berec2007}. 
	Proposed in 2006, \citep{Courchamp2006a} AAE theory argues that price increasing with species rarity induces an Allee threshold for which populations above the threshold are sustainably harvested and populations below the threshold go extinct.
	
	\tab Despite significant traction in the literature \citep[e.g.][]{purcell2014,hinsley2015,harris2013}, there has been no formal analysis validating the existence of an AAE in the original open-access harvest models used to postulate AAE theory. In this paper, we propose a general class of models, which demonstrate that the intuitive arguments for the existence of an AAE \citep{Courchamp2006a,hall2008} can lead to oversimplified conclusions. While previous studies, using alternative descriptions of market dynamics, demonstrate the possibility of low-density equilibria \citep{Auger2010,Clark1990}, population cycles \citep{Burgess2017}, and even extinction \citep{Burgess2017,Ly2014}, our paper reveals previously undiscovered model behavior. 
	
	\tab From a conservation perspective, these trajectories imply that the situation may be more dire than original AAE theory suggests. The model produces a rich set of dynamics with the possibility of classic Lotka-Voltera predator-prey cycles, saddle equilibria, where the Allee population threshold is a function of initial harvest effort, and even cases where populations with abundance far above the anthropogenic Allee threshold are destined to perish. 
	
	\begin{figure}[!h]
		\flushleft
		\begin{subfigure}[t]{0.003\textwidth}
			\textbf{a)}
		\end{subfigure}
		\begin{subfigure}[t]{0.23\textwidth}
			\includegraphics[width=\linewidth,valign=t]{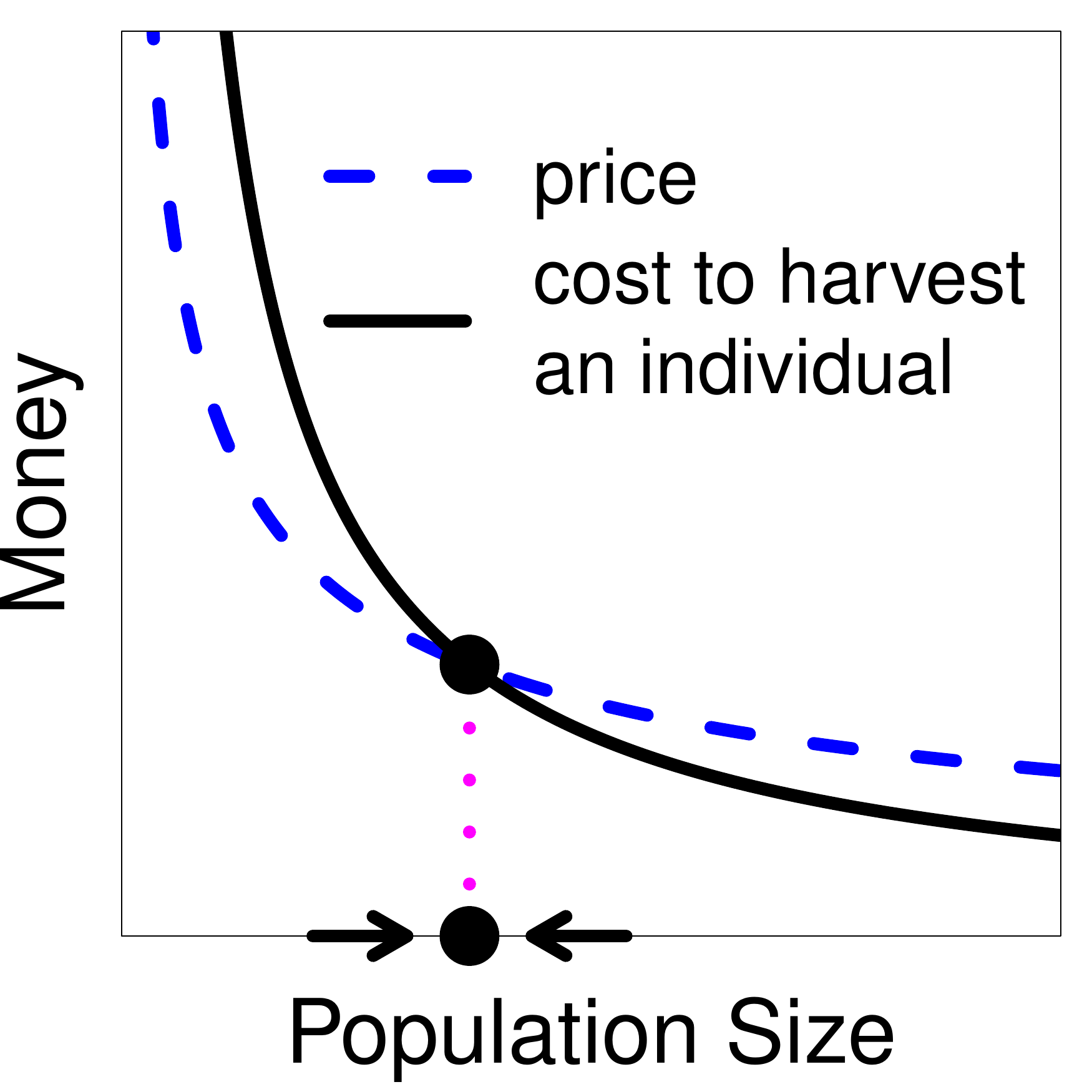}
		\end{subfigure}	
		\begin{subfigure}[t]{0.003\textwidth}
			\textbf{b)}
		\end{subfigure}
		\begin{subfigure}[t]{0.23\textwidth}
			\includegraphics[width=\linewidth,valign=t]{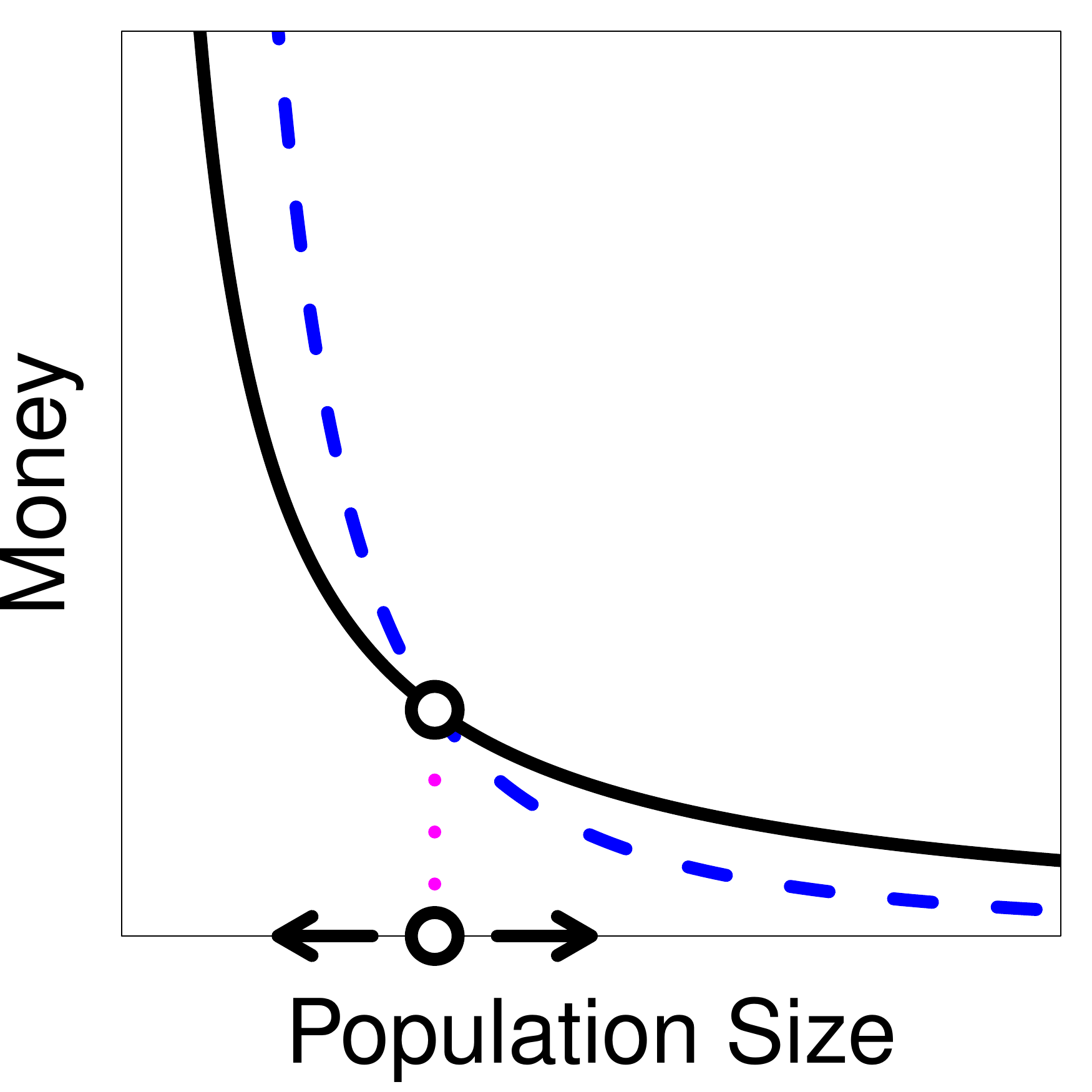}
		\end{subfigure}
		\caption{A graphical representation of the one-dimensional arguments used to propose AAE theory. (a) If the price for a harvested individual is lower than the cost of harvesting that individual when the population is small, but the price is higher than cost when the population is large, then classical arguments claim the population approaches a stable equilibrium (dark circle). (b) If the opposite is true (price is higher than the cost for small populations and lower than the cost for large populations) classic logic suggests that the equilibrium is unstable, creating an Allee threshold. In this paper, we show that these arguments do not always provide correct intuition for the dynamics of harvested populations.}\label{AAEpx}
	\end{figure}
	
	\section{The model}
	
	Consider humans harvesting a population of size $x$, with harvest effort $y$. In the absence of harvest, the population grows at rate $r$. Individuals are harvested at a rate proportional to the product of harvest effort and population density, with catchability coefficient $q$. Harvesters choose to increase their effort if harvest is profitable and decrease their effort if it is unprofitable, which is achieved by letting the change in harvest effort be proportional to the revenue, minus the cost. Assuming the cost of harvest per-unit-effort, $c$, is constant, and the price for harvested individuals, $P(x)$, is a function of population abundance, the above description can be written as,   
	\begin{align}\label{dynamics}
	&\frac{dx}{dt} = rx - qxy\\
	&\frac{dy}{dt} = \alpha\, [\,P(x)qxy - cy\,]\label{market}
	\end{align} 
	
	The parameter $\alpha$ controls the rate at which harvest effort changes with respect to market information. High values of $\alpha$ mean harvesters change effort quickly. We assume all parameters are greater than or equal to zero.
	
	\tab We assume a price abundance relationship,
	\begin{equation}\label{p}
	P(x)=a+b/x^z,
	\end{equation}
	where, $a$ is the minimum price paid per unit harvest when the species is abundant, and $b+a$ is the price when $x=1$. For all $z>0$, price is highest when the species is rare (small $x$). Large $z$ values mean that price decreases more steeply as the population becomes abundant, i.e. increases more steeply with rarity (see Fig \ref{px}). It should be noted that an alternative formulation would have price determined by the rate of individuals being supplied to the market $qxy$ \citep{Burgess2017,Auger2010,Clark1990}. We chose the above function, equation \eqref{propLaequil1}, because it matches the price rarity relationships first sketched by \citet{Courchamp2006a} to formulate the original AAE theory. For a more detailed discussion of this assumption see the section ``Model limitations and assumptions."
	
	\begin{figure}
		\flushleft
		\includegraphics[width=.45\textwidth]{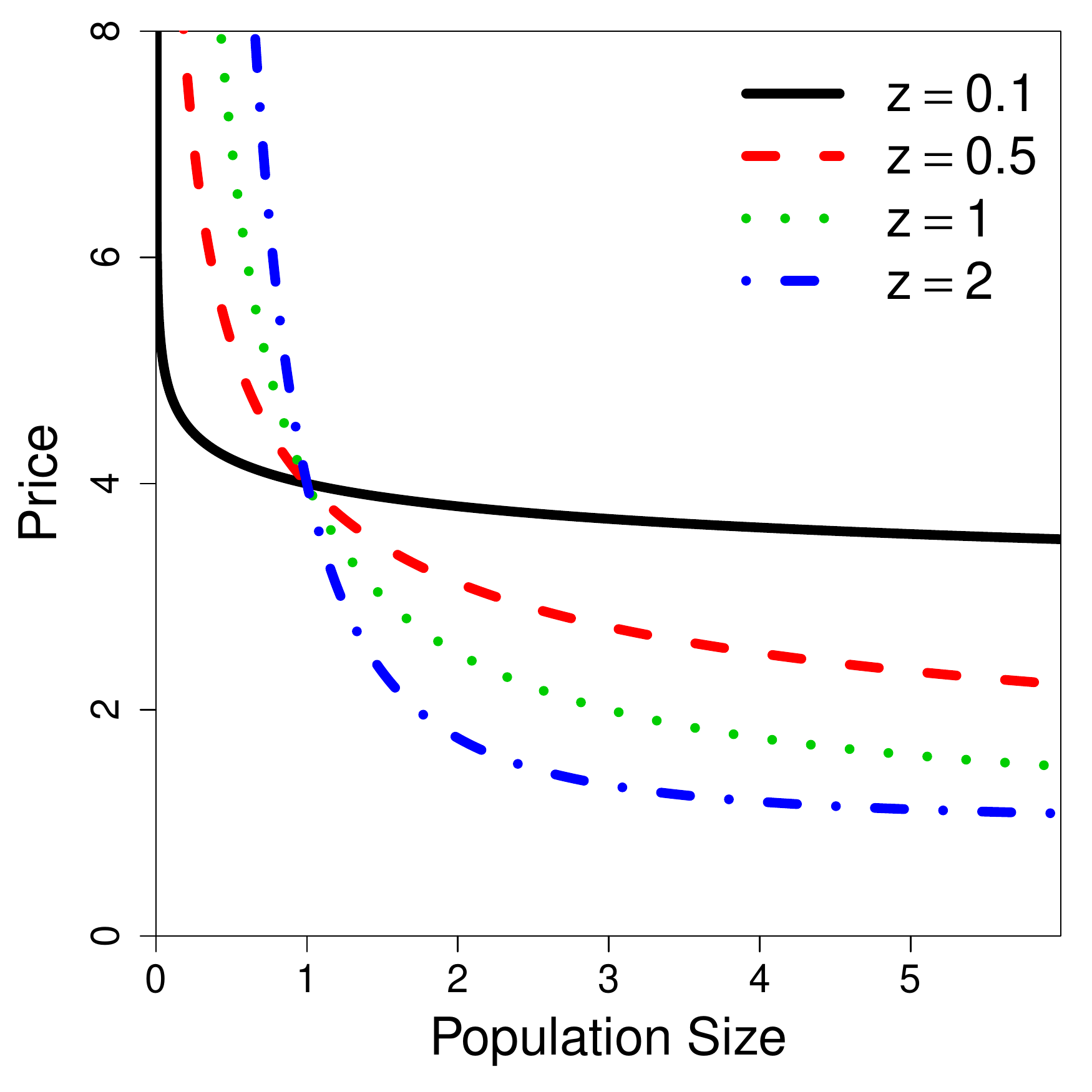}
		\caption{Price as a function of population size for different values of $z$ (sensitivity of price to population size) when $a$ (minimum price consumers pay when the species is abundant) is one, and $b$ (the increased price over $a$ when there is only one individual) is three. Because $a=1$ and $b=3$, when there is only one individual, price is four in all of these curves. This functional form for price can reproduce all of the free-form drawings of price vs population size that were first used to conceptualize the AEE in \citep{Courchamp2006a}.}\label{px}
	\end{figure}
	
	\section{Analysis}  
	\subsection{Lotka-Voltera predator-prey cycles} 
	In the following three sections we assume $a=0$, meaning as the population size approaches infinity, price per unit harvest approaches zero, and then relax this assumption later.
	
	\tab When harvesters receive a constant price $b$, independent of species rarity (i.e. when $z=0$), this model reduces to exactly the classic Lotka-Voltera predator-prey equations from ecology, where harvesters are predators, and the harvested population is prey. But it turns out there is a wide range of price functions that lead to solutions qualitatively identical to classic predator-prey cycles \citep[see][for an introduction to predator-prey models]{Kot2001}. 
	
	\tab There is one positive equilibrium at
	
	\begin{equation}\label{equil}
	x^*=\left(\frac{c}{bq}\right)^{\frac{1}{1-z}}, \, y^*=\frac{r}{q},
	\end{equation}
	
	as long as $z \neq 1$. As an aside, $z=1$, is the case where both price and cost per unit harvest are constant multiples ($b$ and $c/q$ respectively) of $1/x$. If $c/q>b$, per-unit harvest cost is higher than price at all population abundances, and if $c/q<b$, cost is always lower than price. Thus, extinction always occurs if $c<bq$ and persistent growth always occurs if $c>bq$ and the equilibrium above does not exist (as price and cost per unit harvest never cross). 
	
	\tab In the case where price declines less steeply than per unit harvest costs, $0<z< 1$, the system behaves like the classic Lotka-Voltera predation model, with oscillatory population size and harvest effort, with the amplitude of the oscillations depending on initial conditions (Fig. \ref{oscFig}ab). For a proof see Theorem \ref{oscThm} in the Appendix.

	\begin{figure}[!htbp]
		\flushleft
		\begin{subfigure}[t]{0.015\textwidth}
			\textbf{a)}
		\end{subfigure}
		\begin{subfigure}[t]{0.45\textwidth}
			\includegraphics[width=\linewidth,valign=t]{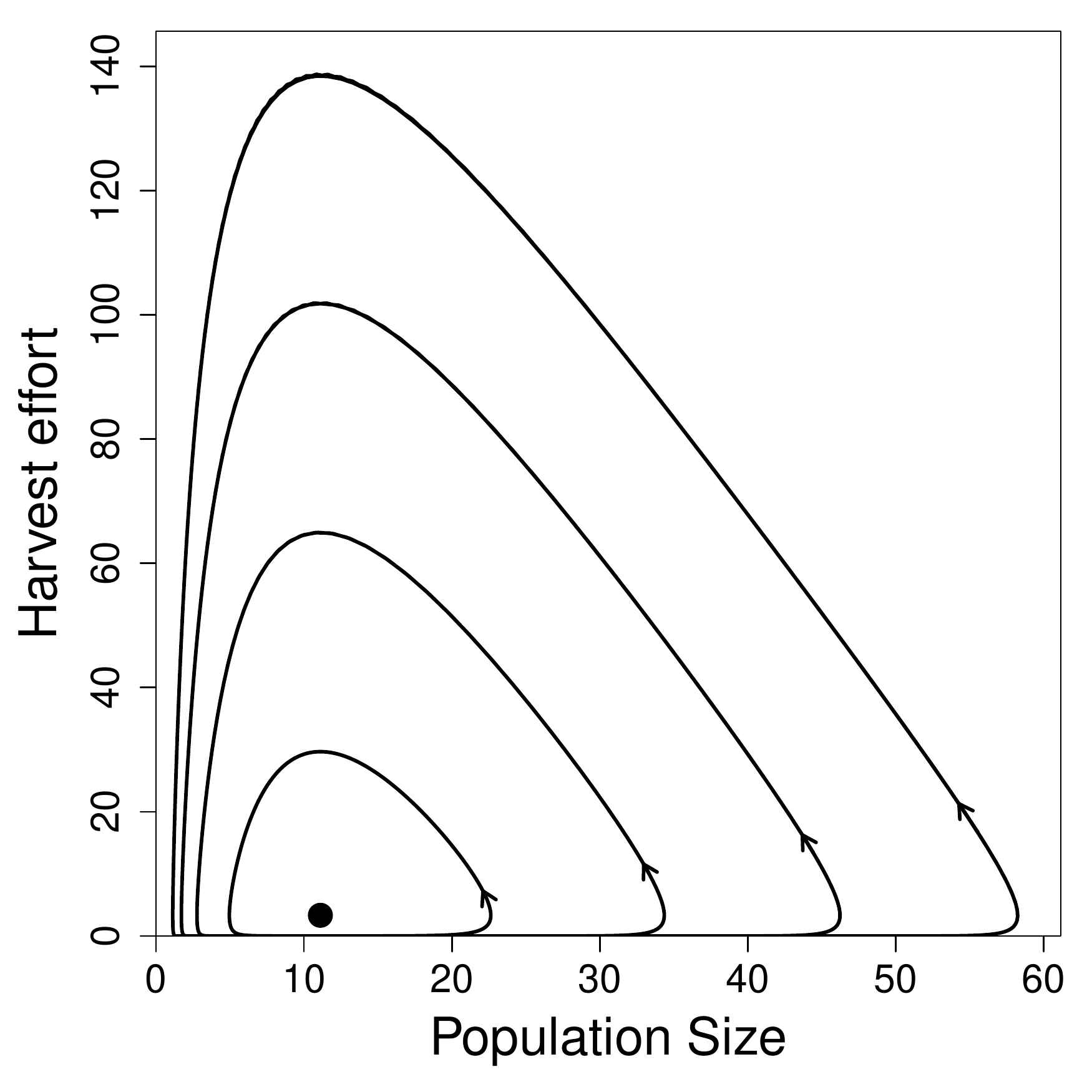}
		\end{subfigure}	
		\begin{subfigure}[t]{0.015\textwidth}
			\textbf{b)}
		\end{subfigure}
		\begin{subfigure}[t]{0.45\textwidth}
			\includegraphics[width=\linewidth,valign=t]{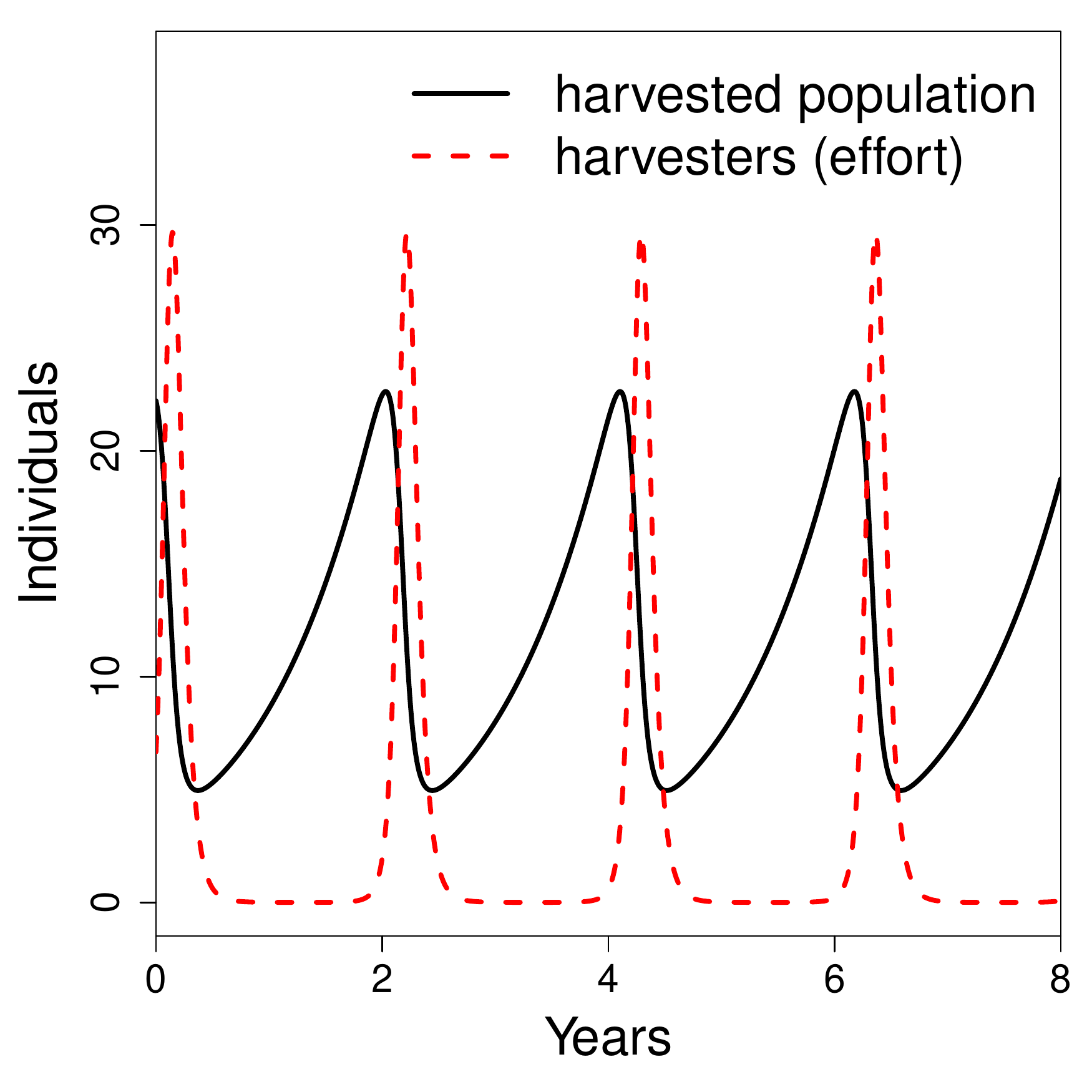}
		\end{subfigure}	\\
		\begin{subfigure}[t]{0.015\textwidth}
			\textbf{c)}
		\end{subfigure}
		\begin{subfigure}[t]{0.45\textwidth}
			\includegraphics[width=\linewidth,valign=t]{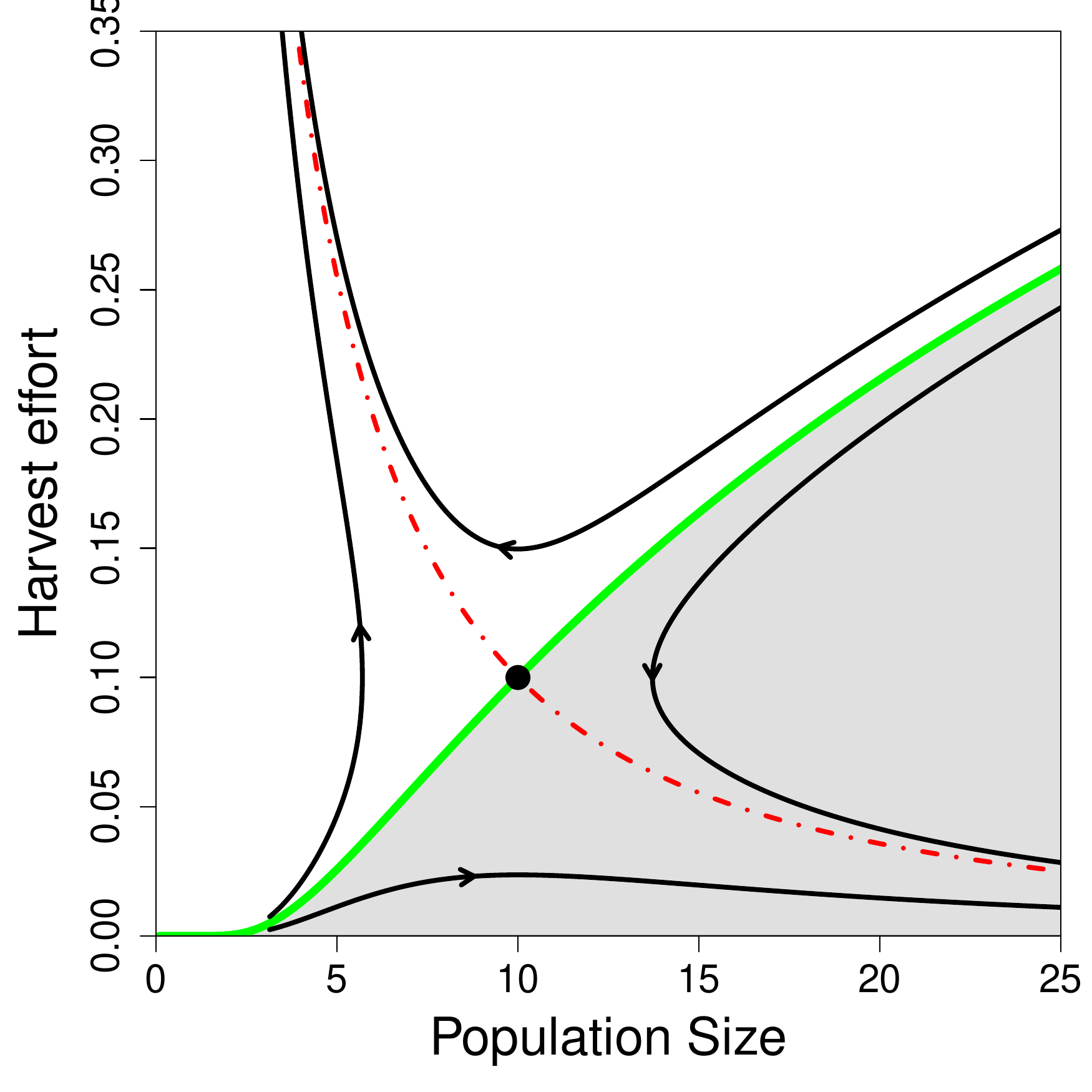}
		\end{subfigure}
		\begin{subfigure}[t]{0.015\textwidth}
			\textbf{d)}
		\end{subfigure}
		\begin{subfigure}[t]{0.45\textwidth}
			\includegraphics[width=\linewidth,valign=t]{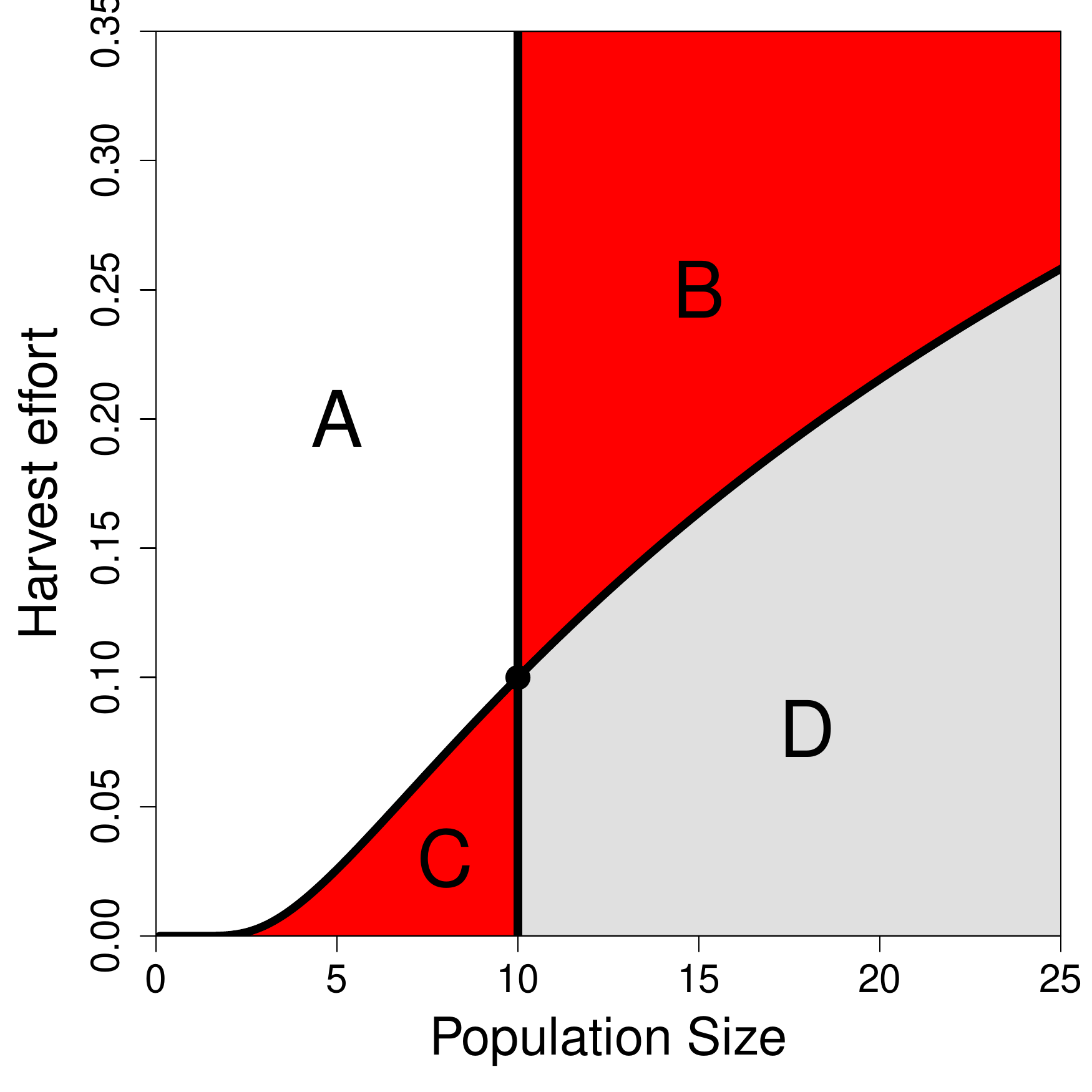}
		\end{subfigure}
		\caption{\textbf{Population dynamics for linear growth model.} (a) Phase portrait for dynamics with $z = 1/2$, and (b) the corresponding time series, from an initial condition of double equilibrium population size and harvest effort [inner ring in (a)]. Other parameters are $r=1, q=.3, c=1, b=1, \alpha=40$. (c) Phase portrait when $z = 2$. The equilibrium of $10$ individuals and harvest effort of $0.1$ is a saddle. The stable manifold (solid green line) determines the fate of the population. All initial conditions below (and to the right) of it (gray shaded area) lead to population growth, and above it (white area) lead to population extinction. (d) Initial harvest effort and population size combinations where one-dimensional arguments (such as those used in classic AAE theory) correctly predicts population extinction, A, and persistence D, and incorrectly predicts population extinction, C, and persistence, B. Other parameters are $r=1, q=10, c=1, b=1, \alpha=2$.}\label{oscFig}
	\end{figure}
	
	\tab It should be noted that the intuitive argument provided in \citep{Courchamp2006a,hall2008} would lead to an incorrect description of the population dynamics in this case. The cost per unit harvest is $c/(qx)$ and the price per unit harvest is $b/x^z$ with $z<1$ (for example $b/\sqrt{x}$), implies that for small populations, cost is greater than revenue and for large populations cost is smaller than revenue (see Fig \ref{AAEpx}a). Intuitively, one might falsely conclude that the equilibrium is stable (see Fig \ref{AAEpx}a), when indeed, we have shown that it is a center and hence not stable (Fig. \ref{oscFig}ab). 
	
	\tab It should also be noted that this model is undefined when $x=0$, for all $z>0, z\neq1$, and hence why $x=0, y=0$ is not an equilibrium, as in many other population models.

	\subsection{Anthropogenic Allee effect} 
	If price declines more steeply with respect to population size than per unit harvest costs, $z>1$, there are two possible long-term outcomes, either the population 1) grows to infinity (e.g. grey area in Fig. \ref{oscFig}c) or 2) crashes to extinction (e.g. the white area in Fig. \ref{oscFig}c). The fate of the population between extinction and long-term growth depends both on the initial population size and initial harvest effort. Note that unlike the intuitive arguments described by \citet[see Fig. \ref{AAEpx}b and the original figures in][]{Courchamp2006a}, initial harvest effort plays an important role. There is no fixed population size for which, if the population starts above that size it always persists, and if it starts below that size, it always goes extinct (a concept the authors coin an \textit{anthropogenic Allee effect} [AAE]). 
	
	\tab If the population size starts below the equilibrium level $x^*$ (as given in equation \eqref{equil}) and initial harvest effort is small, the population can still persist (initial conditions in region C, Fig \ref{oscFig}d). The opposite is also possible; if the population size starts above the threshold $x^*$ and initial harvest effort is high, the population can still go extinct (initial conditions in region B in Fig \ref{oscFig}d). The reason for this is that for $z>1$, the equilibrium \eqref{equil}, is a saddle (a two-dimensional equilibrium that does not have the same qualitative behavior as an unstable equilibrium in an analogous one-dimensional system). It has both a stable and unstable manifold. The stable manifold acts as a ``separatrix," a curve in two-dimensional, $(x,y)$ space, separating populations destined for extinction and those that will survive. For all initial harvest efforts below the stable manifold, the population approaches the branch of the unstable manifold which goes to infinity.  For all harvest efforts above the stable manifold the population goes extinct, approaching the upper branch of the unstable manifold (Fig. \ref{oscFig}c). 
	
	\tab The slope of the separatrix, at the saddle equilibrium, can be solved analytically (see Appendix for derivation), as
	
	\begin{equation}\label{slope}
	b\sqrt{\frac{\alpha r (z-1) }{c} }
	\left(\frac{c}{bq}\right)^{\frac{z}{z-1}}.
	\end{equation}
	
	For all $z>1$, It is easy to show that this expression increases with respect to $\alpha$, $r$, and $c$ and decreases in $b$ and $q$ (see Appendix for proof). The dependence on $z$ is more complicated than for the other parameters, but the slope increases with $z$ as long as $c<bq$. This means that for high values of population growth rate, harvest cost, and harvest effort adjustment rate, the separatrix becomes a near vertical line, such that no matter the initial harvest effort, small populations to the left of the line are destined to extinction, and large populations destined to long-term growth (as is the case in classic AAE theory). For low values of these parameters, the fate of the population is more heavily influenced by initial harvest effort. The opposite is true for the other parameters (see Fig. \ref{steepParms} in the appendix for plots of how the slope in equation \eqref{slope} changes with respect to all the parameters). As an example of how the separatrix changes with respect to $\alpha$, see Fig. \ref{steep} in the Appendix as well.

	\subsection{Density dependence}
	
	We now examine the case of density-dependent growth, the same as \eqref{dynamics} but where the first term in the population growth equation is replaced with logistic growth to carrying capacity $k$, rather than linear growth
	
	\begin{align}\label{NLdynamics}
	&\frac{dx}{dt} = rx(1-x/k) - qxy
	\end{align} 
	
	This is the starting point for the majority of work on the effect of price dynamics on harvesting \cite{Clark1990,Auger2010,Colin2010,Mansal2014,Ly2014}. However, we will show that the linear system \eqref{dynamics} is a good approximation for the dynamics of the nonlinear system in many scenarios.
	
	\tab For $z=0$ the model reduces to one previously studied for predator-prey populations \citep{Kot2001} and open-access fisheries \citep{bjorndal1987}, and has been shown to have a stable positive equilibrium  \citep{Kot2001,bjorndal1987}.
	
	\tab For all parameters, there is an equilibrium at  $x=k$, $y=0$ (no harvest). This carrying capacity equilibrium is an unstable saddle if $c<bqk^{1-z}$, meaning if the cost of harvest is lower than the price received when the population is at carrying capacity, some harvest will occur, and populations will decline to levels below carrying capacity. If $c>bqk^{1-z}$, the equilibrium is stable, as harvesting populations at carrying capacity is unprofitable (see Appendix for proof). 
	
	\tab for $z>0$ with $z\neq1$, there is an additional equilibrium at:
	
	\begin{equation}\label{nlequil}
	x^*=\left(\frac{c}{bq}\right)^{\frac{1}{1-z}}, \, y^*=\frac{r}{q}\left(1-\frac{x^*}{k}\right).
	\end{equation}
	
	This equilibrium has the same population size as in the linear case. However, equilibrium harvest effort is reduced by the proportion that $x$ is below carrying capacity, $1-x^*/k$. 
	
	\tab Similar to the linear system, for $0<z<1$, the dynamics are oscillatory (Fig. \ref{nlFig} ab). However, the positive equilibrium, in equation \eqref{nlequil}, is stable, and all initial conditions spiral in towards it (Fig. \ref{nlFig} ab). Note though, that if the equilibrium is far below carrying capacity, the approach to equilibrium is slow, with large amplitude oscillations (which for short management timescales are effectively as in the linear system) (Fig. \ref{nlFig} a).
	
	\tab For $z=1$ the dynamics behave similarly to the linear system. Populations either approach carrying capacity if $c>bq$ or decline to zero otherwise.
	
	\tab For $z>1$ the dynamics are also similar to the linear case (Fig. \ref{nlFig}c). If $c>bqk^{1-z}$, which implies $x^*$, in equation \eqref{nlequil}, is below carrying capacity, then this equilibrium is an unstable saddle (see Appendix for proof). The population either declines to zero or approaches carrying capacity, depending on the initial population size and harvest effort. If $z>1$ and $c<bqk^{1-z}$ then the population always declines to zero, since the nontrivial equilibrium induced by market dynamics is higher than carrying capacity (see appendix for proof).
	
	\begin{figure}[!htbp]
		\flushleft
		\begin{subfigure}[t]{0.015\textwidth}
			\textbf{a)}
		\end{subfigure}
		\begin{subfigure}[t]{0.45\textwidth}
			\includegraphics[width=\linewidth,valign=t]{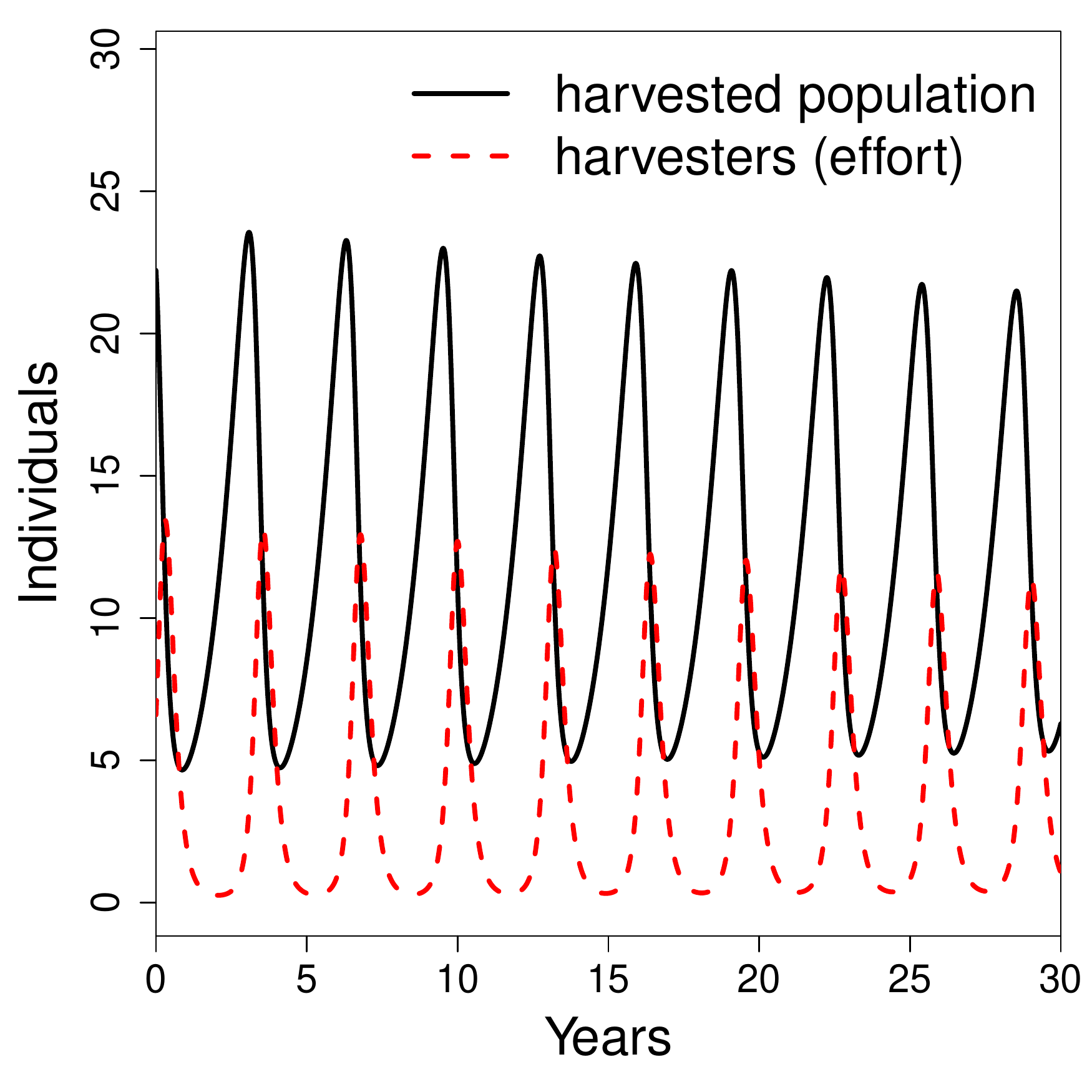}
		\end{subfigure}	
		\begin{subfigure}[t]{0.015\textwidth}
			\textbf{b)}
		\end{subfigure}
		\begin{subfigure}[t]{0.45\textwidth}
			\includegraphics[width=\linewidth,valign=t]{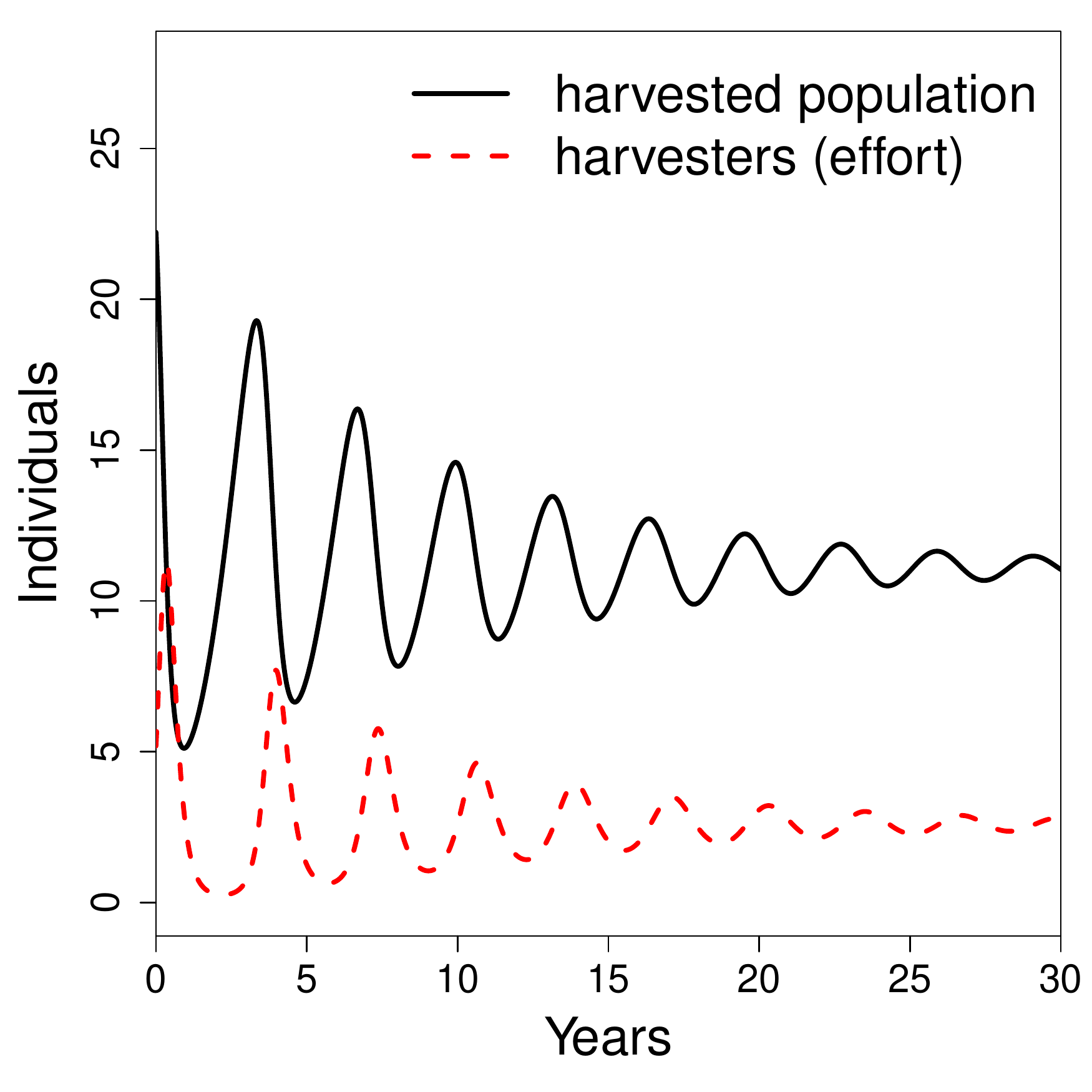}
		\end{subfigure}	\\
		\begin{subfigure}[t]{0.015\textwidth}
			\textbf{c)}
		\end{subfigure}
		\begin{subfigure}[t]{0.45\textwidth}
			\includegraphics[width=\linewidth,valign=t]{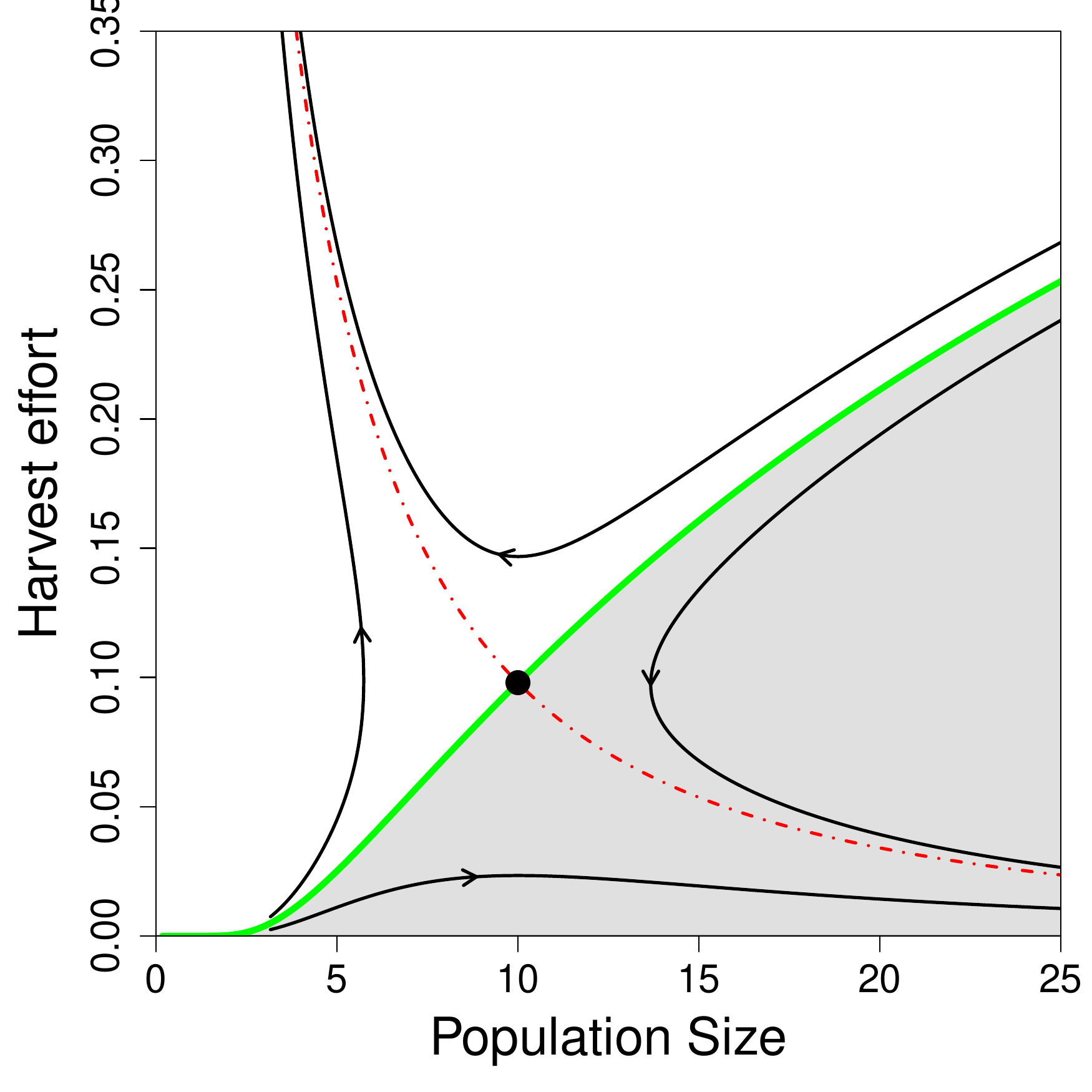}
		\end{subfigure}
		\begin{subfigure}[t]{0.015\textwidth}
			\textbf{d)}
		\end{subfigure}
		\begin{subfigure}[t]{0.45\textwidth}
			\includegraphics[width=\linewidth,valign=t]{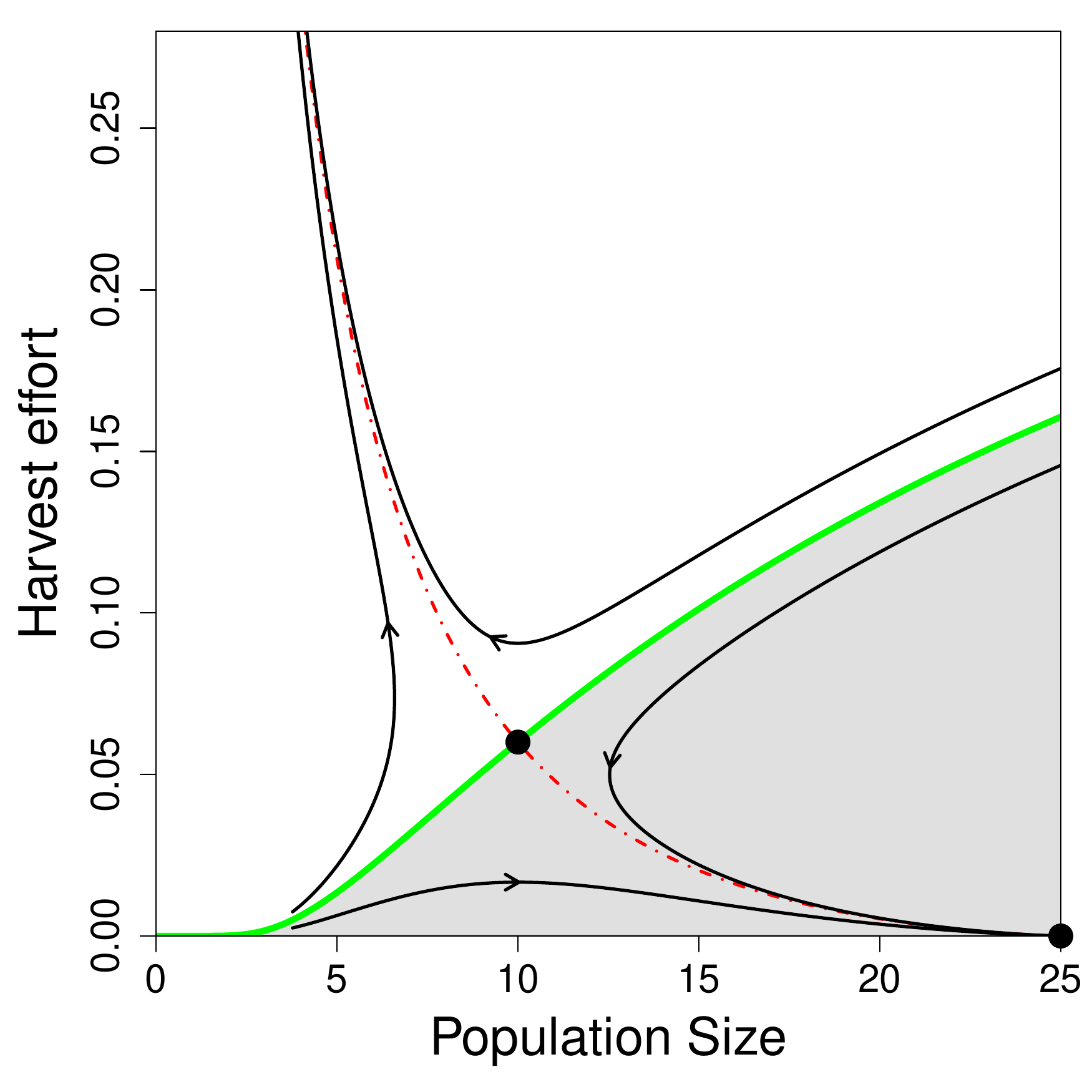}
		\end{subfigure}
		\caption{\textbf{Dynamics under density dependent growth}. (ab) Trajectories for the baseline parameterization with $z=1/2$ and the added logistic nonlinearity for (a) carrying capacity $k=1,000$, and (b) $k=50$. Other parameters are $r=1, q=.3, c=1, b=1, \alpha=10$. If the equilibrium is far below carrying capacity, $k=1,000$, its stability is weak and the dynamics are similar to the linear case (Fig. \ref{oscFig}b), when the equilibrium is close to carrying capacity, $k=50$, the approach to the equilibrium is faster. (c) The phase plane, with $z=2$ and $k=500$, which behaves the same as in the linear case. Other parameters are $r=1, q=10, c=1, b=1, \alpha=2$. (d) When $k=25$ in the dynamics are qualitatively similar, but with reduced poaching effort at equilibrium and along the stable manifold (green).}\label{nlFig}
	\end{figure}
	
	\tab The stability of the $x=k, y=0$ equilibrium is consistent with the above dynamics. In the case where $z>1$ and $x^*$ is a saddle, $x=k, y=0$ equilibrium is locally stable, and unstable when $c<bqk^{1-z}$ (see Appendix). This just says that when there is a positive saddle equilibrium with $x^*<k$, as long as the initial condition starts close enough to $x=k, y=0$, solutions will approach $x=k, y=0$. For the oscillatory case, $z<1$, $x=k, y=0$, is an unstable saddle if $x^*<k$. This is consistent with the corresponding oscillatory dynamics in Fig. \ref{nlFig}ab.
	
	\subsection{Crossing the Allee threshold when there is a minimum price}
	Consider the original system \eqref{dynamics} with a non-zero minimum price, $a$, regardless of species abundance (i.e. price does not go to zero as the species becomes very common). This makes the model more difficult to analyze. However, the equilibria can be computed analytically for $z=1/2, 1$ and $2$. 
	
	\tab When $z=1$, and if $c>bq$, there is an equilibrium at 
	
	\begin{equation}
	x=\frac{c-bq}{aq}, \, y=\frac{r}{q},
	\end{equation}
	
	and this equilibrium is stable. Otherwise, this equilibrium does not exist, and the population grows indefinitely.
	
	\tab For $z=2$, if $a>c^2/(4bq^2)$, the price per unit harvested is always greater than the cost, and therefore harvest goes to infinity as the population goes extinct. If $a<c^2/(4bq^2)$, there are two equilibria,
	\begin{align}
	x_1 &=\frac{c-\sqrt{c^2-4abq^2}}{2aq}, \, y_1=\frac{r}{q},\label{laequil1}\\
	x_2 &=\frac{c+\sqrt{c^2-4abq^2}}{2aq}, \, y_2=\frac{r}{q}.\label{laequil2}
	\end{align}
	Linearization (see Appendix) shows that the smaller equilibrium, \eqref{laequil1}, is always an unstable saddle, confirming the similarity of the dynamics in this case to the dynamics exhibited by our simple model (as in Figs. \ref{oscFig}c and \ref{nlFig}c).
	
	\tab However, the unstable and stable manifold of this saddle connect, forming a homoclinic orbit surrounding the larger equilibrium, which is a center (Fig. \ref{laFig}a). This is a dire result because it means that only populations with initial conditions located inside the homoclinic orbit will persist, as all other initial conditions will lead to population extinction. This includes large populations far above the classic Allee threshold (the smaller equilibrium) and hence proposes a potential mechanism for how abundant species cross the Allee threshold on the way to extinction. 
	
	\tab For $z=1/2$, there is one equilibrium at
	\begin{align}
	x&=\frac{2ac+b^2q-b\sqrt{4acq+b^2q^2}}{2a^2q}, \, y=\frac{r}{q},
	\end{align}
	for which the trace of the Jacobian is zero, and simulation suggests that the equilibrium is indeed a center (Fig. \ref{laFig}b), just as in the case for $z<1$ for in the simple model, as written in equation \eqref{dynamics}.
	
	\begin{figure}[!htb]
		\flushleft
		\begin{subfigure}[t]{0.015\textwidth}
			\textbf{a)}
		\end{subfigure}
		\begin{subfigure}[t]{0.45\textwidth}
			\includegraphics[width=\linewidth,valign=t]{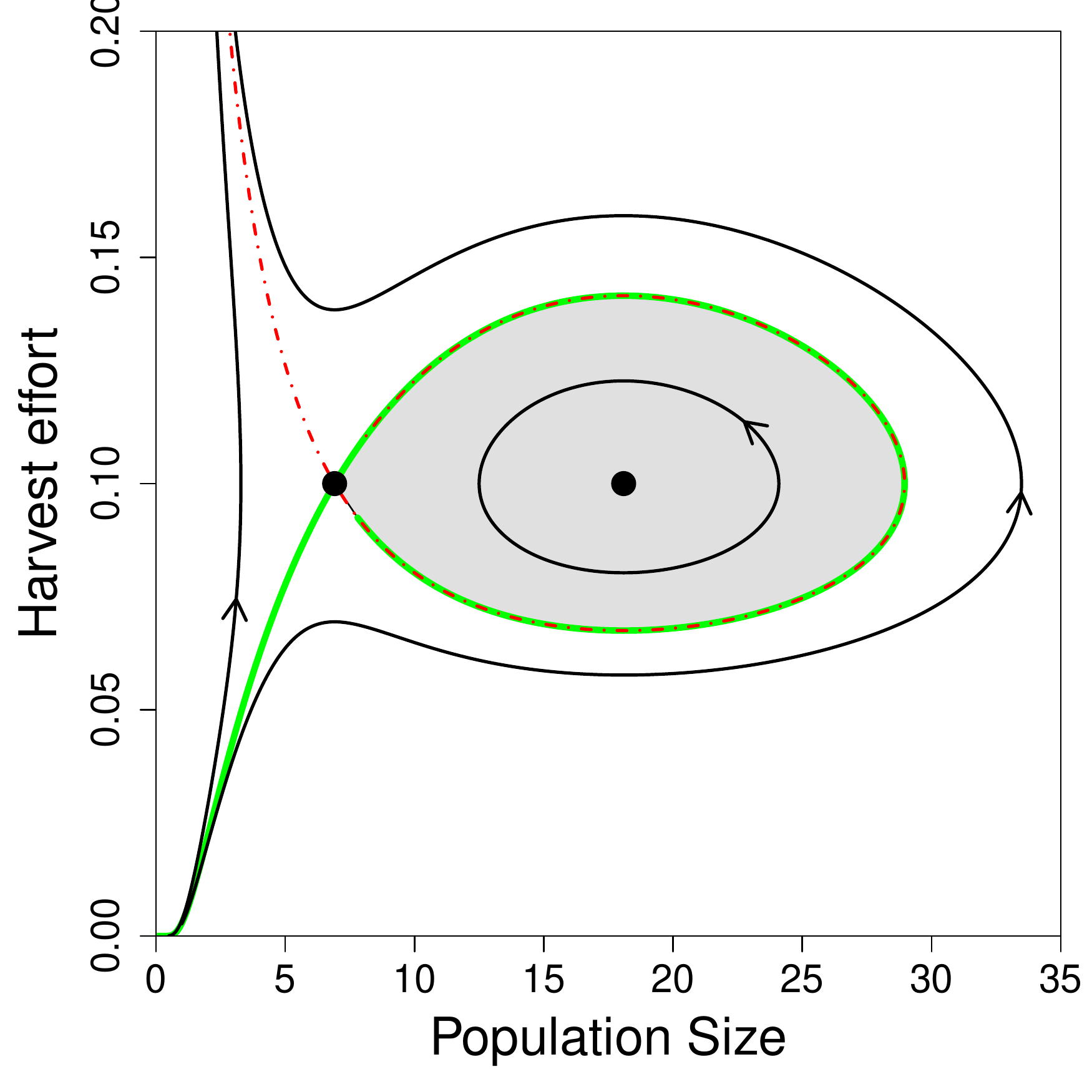}
		\end{subfigure}	
		\begin{subfigure}[t]{0.015\textwidth}
			\textbf{b)}
		\end{subfigure}
		\begin{subfigure}[t]{0.45\textwidth}
			\includegraphics[width=\linewidth,valign=t]{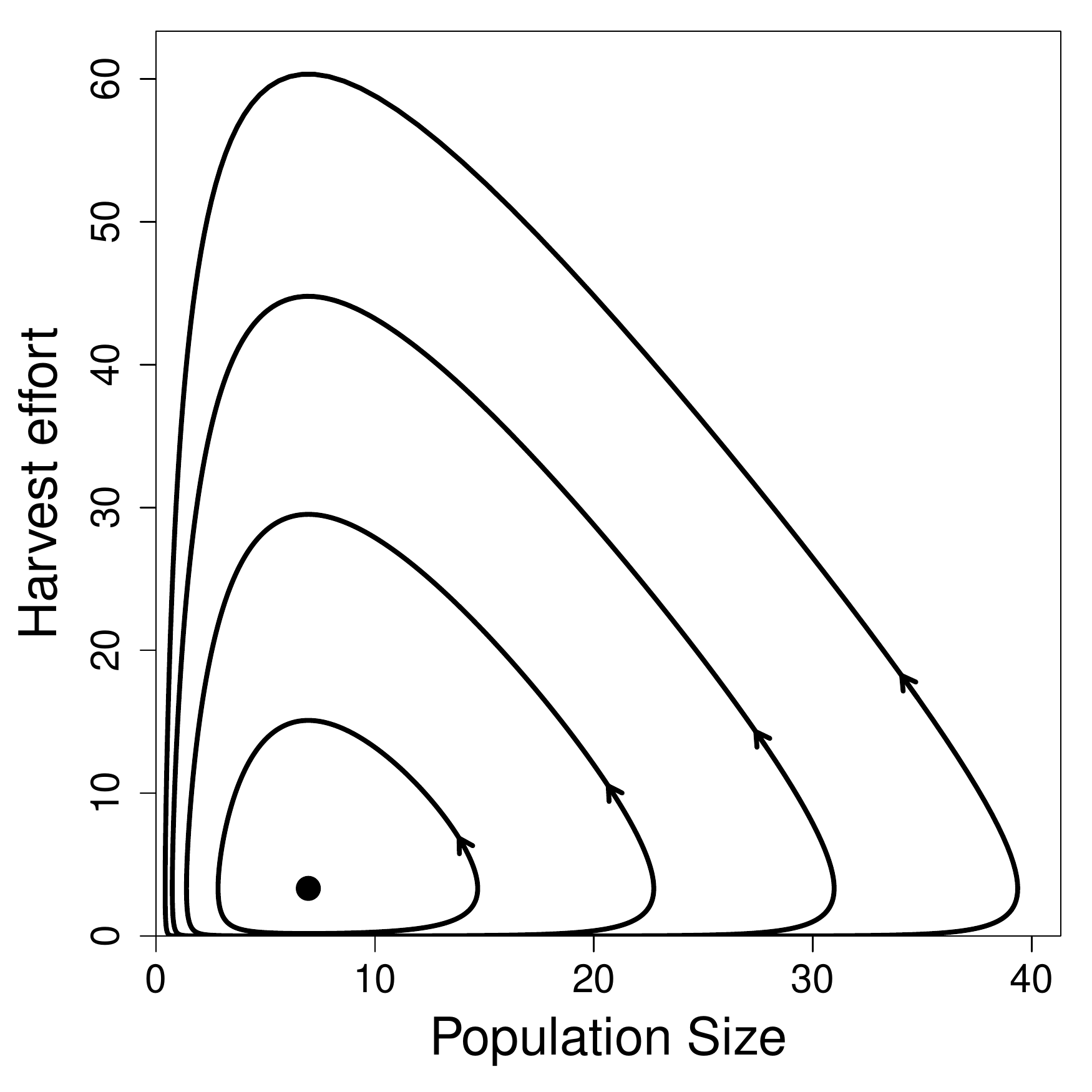}
		\end{subfigure}
		\caption{\textbf{Dynamics when growth is linear and price is always above a fixed value, $a$.} (a) When $z=2$, trajectories that start in the grey area oscillate around the larger equilibrium population size, while trajectories outside of the shaded region eventually approach extinction. The green solid line is the stable manifold of the saddle equilibrium and the red dotted line is the unstable manifold. Other parameters are $q=10, r=1, c=1, b=.5,\alpha = 1$ and $a=0.04$. When $z=1/2$ (b), populations oscillate indefinitely as they do in Fig. \ref{oscFig}ab. Other parameters are $q=0.3, r=1, c=1, b=1, \alpha=10$ and $a=0.1$.}\label{laFig}
	\end{figure}
	
	\tab Combining both density dependence and such a general price abundance relationship makes an analytic approach difficult. The $x=k, y= 0$ equilibrium still exists in this model and is stable if $c>ka+bqk^{1-z}$. But the other equilibria are not easily solved for analytically, even for special cases of $z$.
	
	\tab Numerically, we demonstrate that the logistic perturbation to the linear population growth term breaks the homoclinic orbit and the larger equilibrium becomes stable (Fig. \ref{nlaFig}). Once again if $k$ is large compared to equilibrium population size, then the dynamics are similar to the linear system with a homoclinic orbit (Fig. \ref{nlaFig}bd), and when $k$ is small, the approach to equilibrium can be fast (Fig. \ref{nlaFig}ab). If harvesters adjust effort quickly, the area for which the population persists stretches vertically  (compare Fig. \ref{nlaFig}ab to \ref{nlaFig}cd).
	
	\begin{figure}[!htbp]
		\flushleft
		\begin{subfigure}[t]{0.015\textwidth}
			\textbf{a)}
		\end{subfigure}
		\begin{subfigure}[t]{0.45\textwidth}
			\includegraphics[width=\linewidth,valign=t]{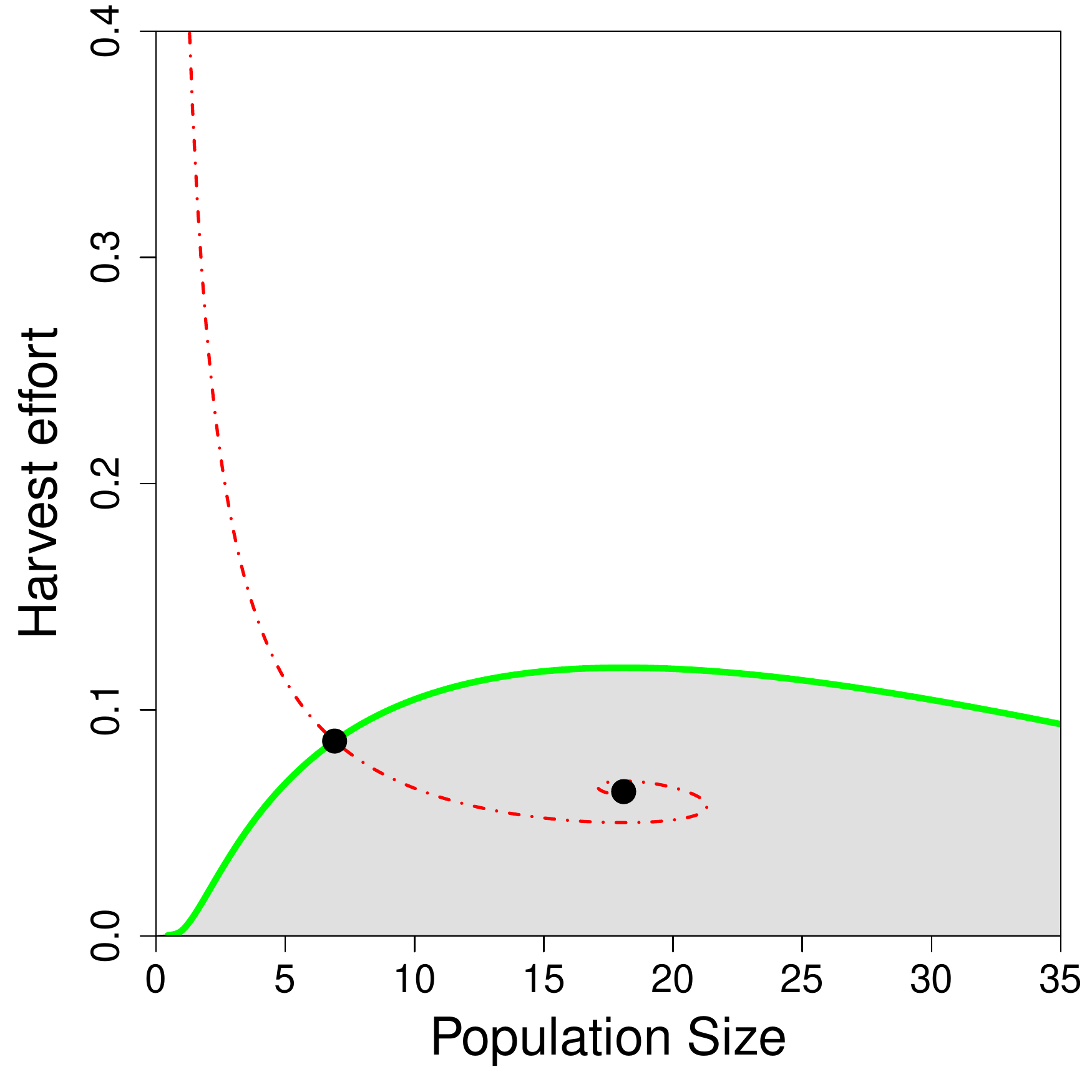}
		\end{subfigure}	
		\begin{subfigure}[t]{0.015\textwidth}
			\textbf{b)}
		\end{subfigure}
		\begin{subfigure}[t]{0.45\textwidth}
			\includegraphics[width=\linewidth,valign=t]{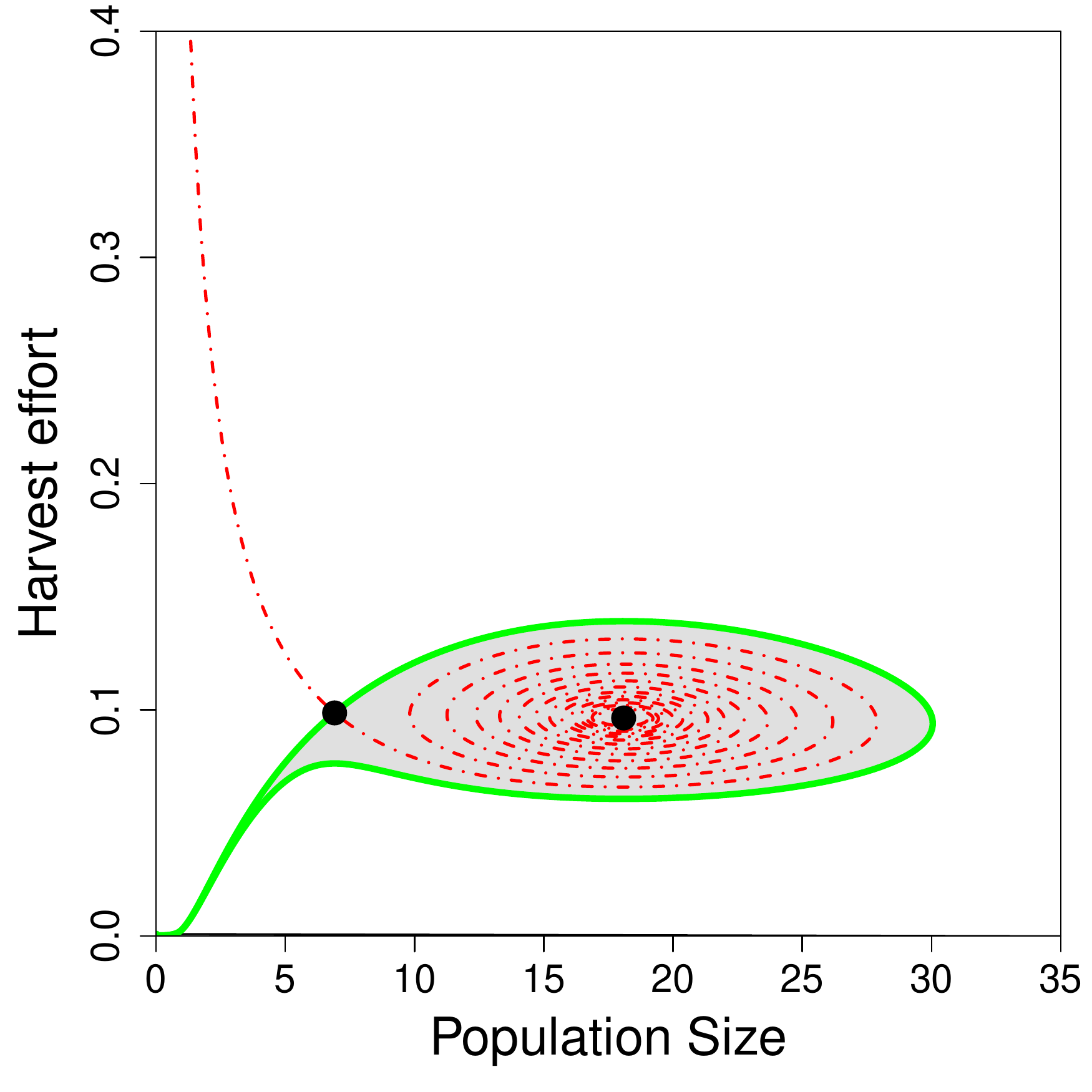}
		\end{subfigure}	\\
		\begin{subfigure}[t]{0.015\textwidth}
			\textbf{c)}
		\end{subfigure}
		\begin{subfigure}[t]{0.45\textwidth}
			\includegraphics[width=\linewidth,valign=t]{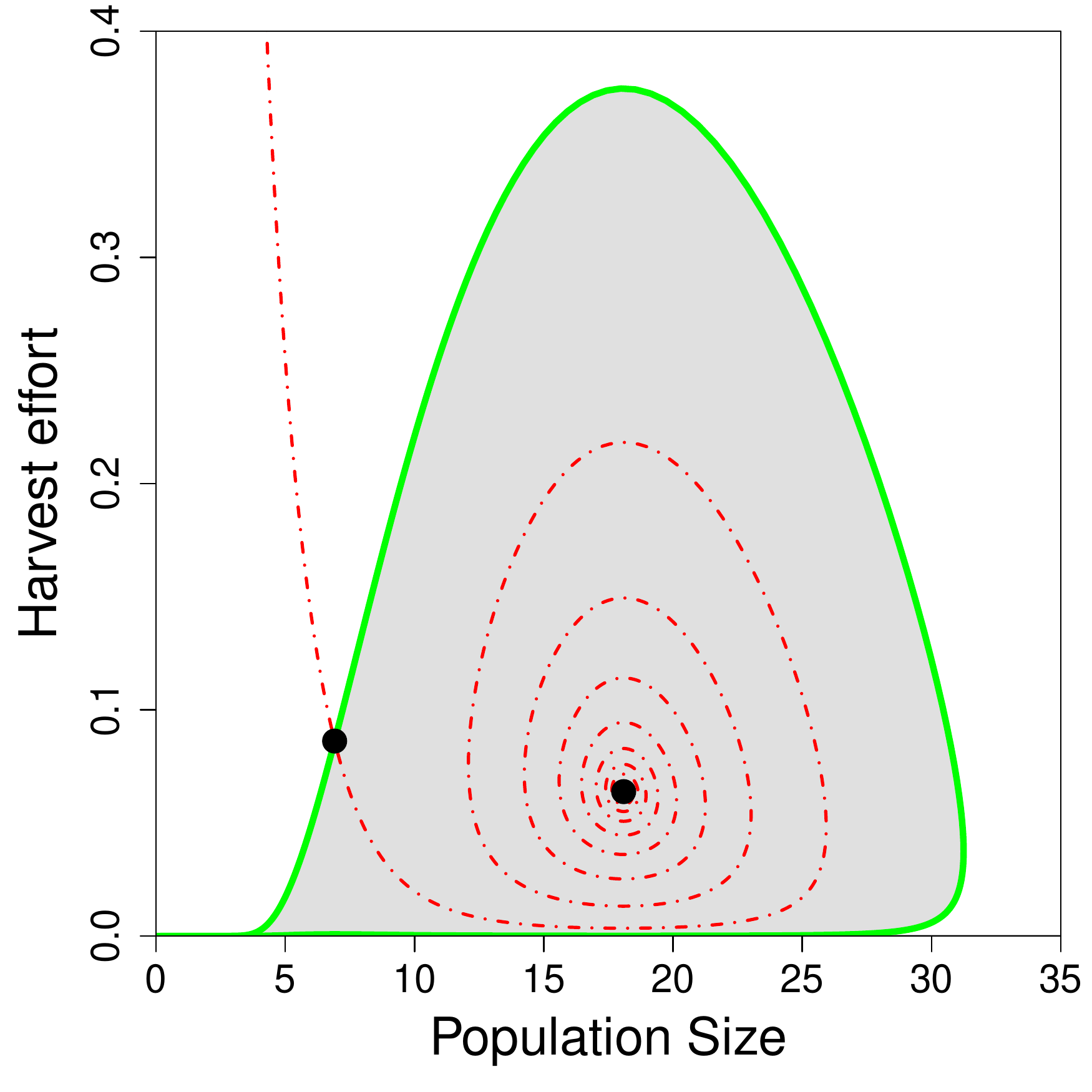}
		\end{subfigure}
		\begin{subfigure}[t]{0.015\textwidth}
			\textbf{d)}
		\end{subfigure}
		\begin{subfigure}[t]{0.45\textwidth}
			\includegraphics[width=\linewidth,valign=t]{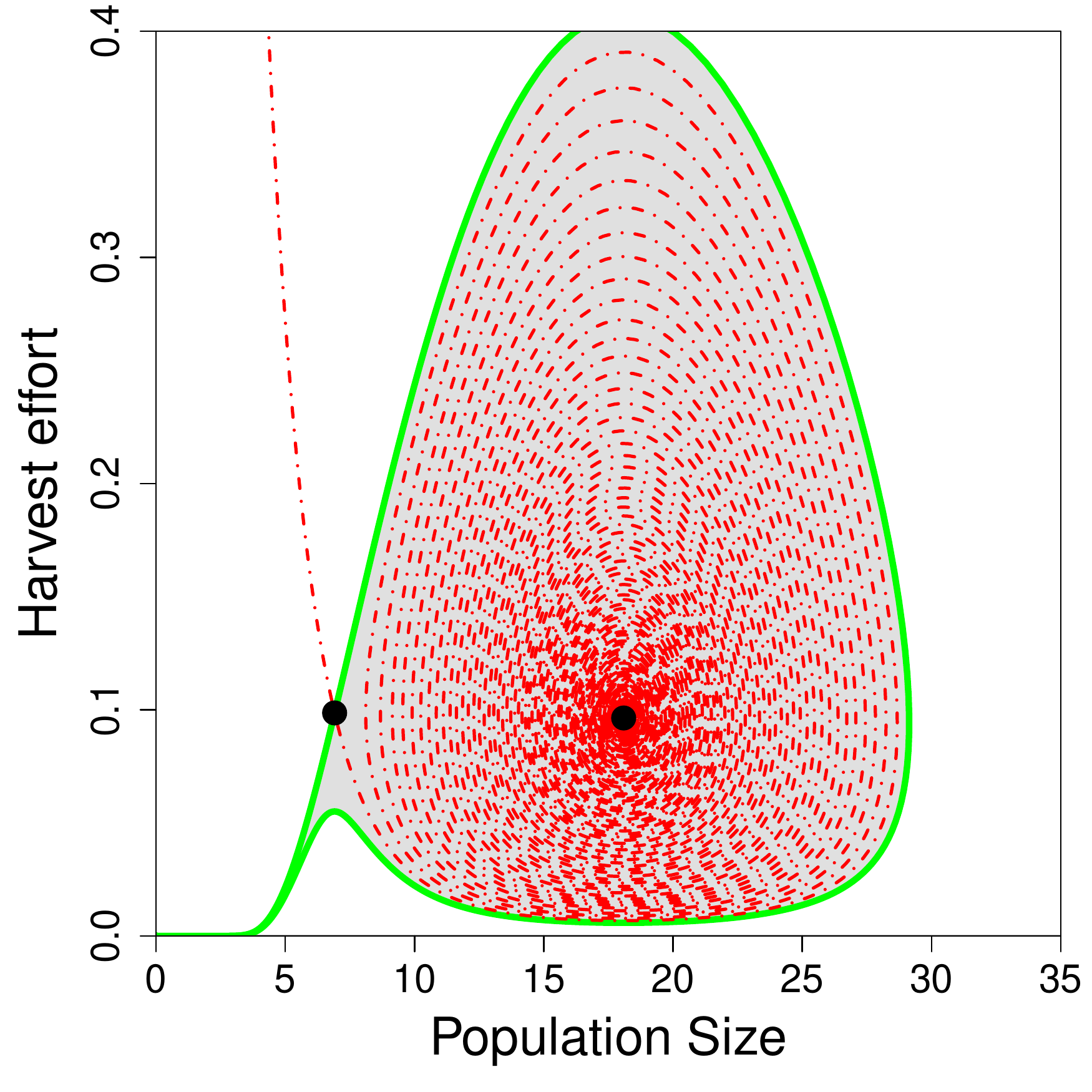}
		\end{subfigure}
		\caption{\textbf{Dynamics when growth is logistic and price is always above a fixed value} with $z=2$. The stable (solid, green) and unstable (dashed, red) manifold are displayed for (a) $\alpha=1$ and $k=50$ (b) $\alpha=1$ and $k=500$ (c) $\alpha=25$ and $k=50$ and (d) $\alpha=25$ and $k=500$. Other parameters are $r = 1,	q = 10, c = 1, b = .5$ and $a = 0.004$. The grey area indicates all initial combinations of population size and harvest effort that lead to population persistence, and the white areas indicate extinction. For high values of $k$ (bc) the population behaves nearly as in the linear case, but the homoclinic orbit is broken and initial conditions that start in the gray area spiral into the stable equilibrium, but very slowly.}\label{nlaFig}
	\end{figure}
	
	\tab For $0<z<1$ the dynamics (Fig. \ref{nlaOscFig}) are similar to the corresponding linear systems  (Fig. \ref{oscFig} and \ref{laFig}b).

	\begin{figure}[!htb]
		\flushleft
		\begin{subfigure}[t]{0.015\textwidth}
			\textbf{a)}
		\end{subfigure}
		\begin{subfigure}[t]{0.45\textwidth}
			\includegraphics[width=\linewidth,valign=t]{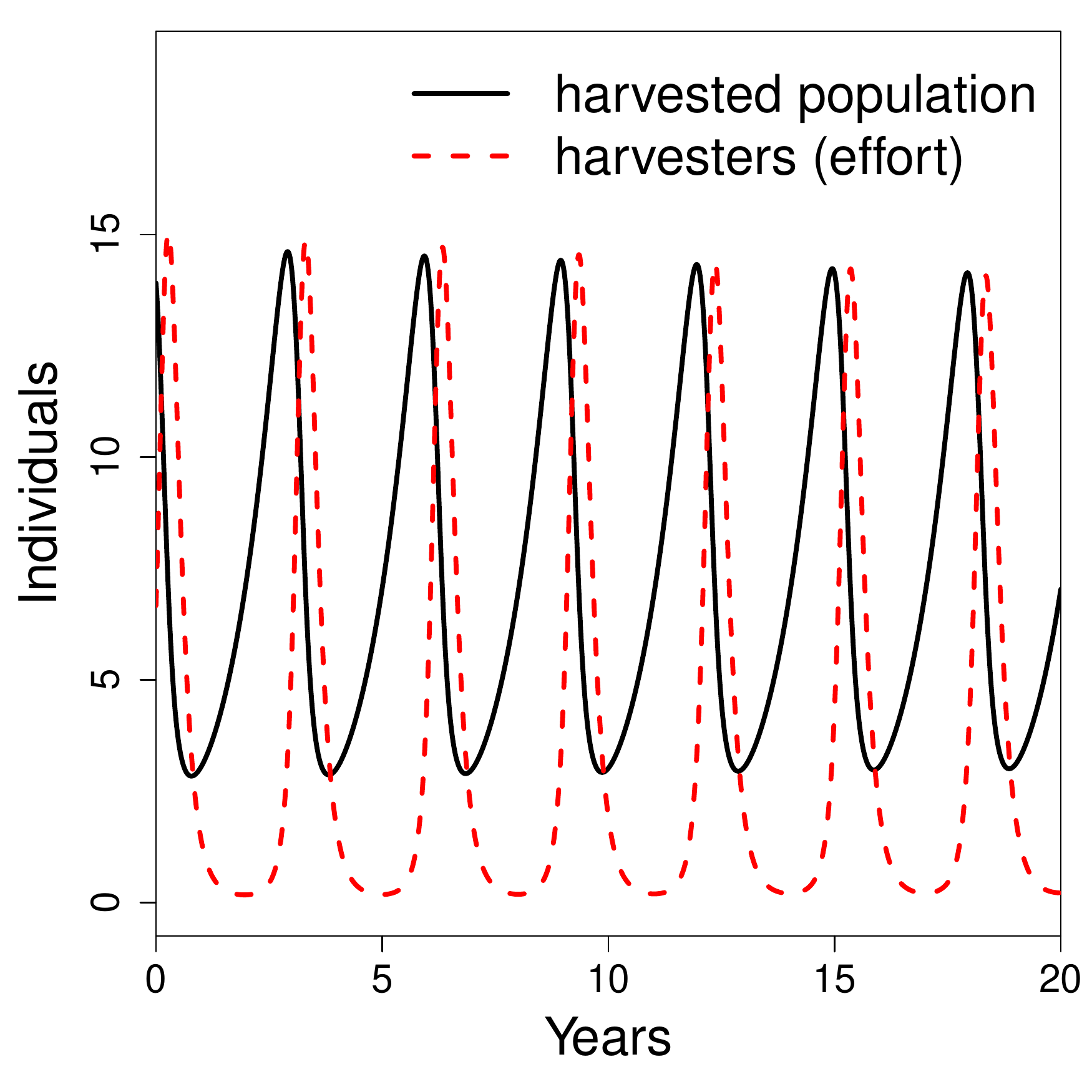}
		\end{subfigure}	
		\begin{subfigure}[t]{0.015\textwidth}
			\textbf{b)}
		\end{subfigure}
		\begin{subfigure}[t]{0.45\textwidth}
			\includegraphics[width=\linewidth,valign=t]{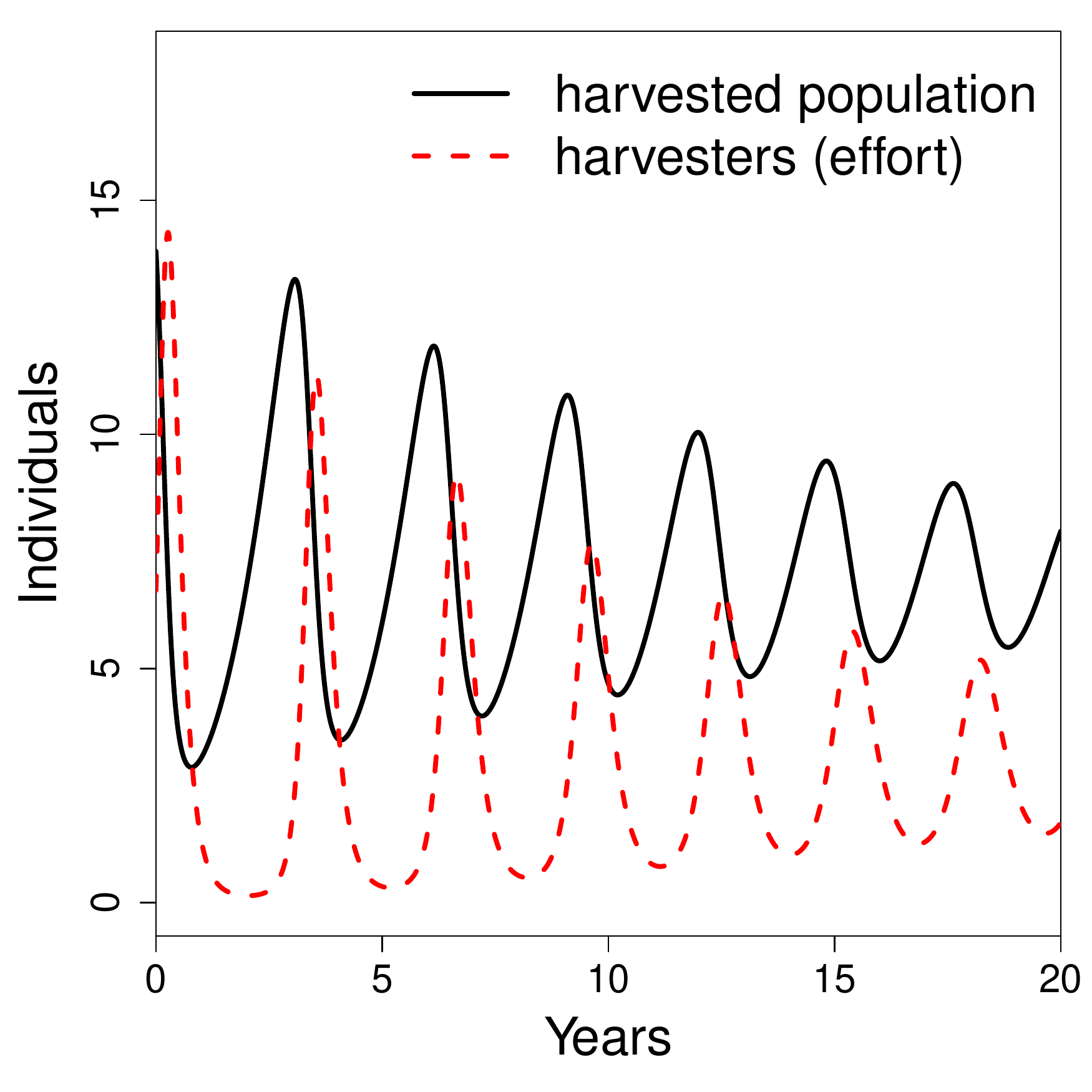}
		\end{subfigure}
		\caption{\textbf{Dynamics when growth is logistic, price is always above a fixed value, and $z=1/2$.}  Population size (solid, black) and harvest effort (dashed, red) through time for (a) $k=1,000$ and (b) $k=50$. Other parameters are $r = 1, q = 0.3, c = 1, b = .5, \alpha=10$ and $a = 0.1$. Population always oscillates but approaches a stable equilibrium. The approach is quick to equilbrium for small $k$ (b) but slow, similar to the linear growth case, for large $k$ (a).}\label{nlaOscFig}
	\end{figure}	
	
	\section{Numerical example: elephant poaching}
	Elephants have been in rapid decline over the last eight years due to increased poaching for ivory \cite{Gross2016}, and therefore the international community is quickly implementing new policies to reduce this illegal harvest \cite[e.g.][]{Biggs2016}. However, to determine which policies are likely to work best, we must understand how poaching and elephant population abundance will respond to different price-abundance relationships. To do this we parameterize our model for African elephant populations illegally harvested for ivory, using the estimates \citep[$r=0.06, c=2000, k=500,000, x_0=300,000$, reported in][]{Lopes2015a}, except for $q$, which we set based on 3.54 elephants killed per poaching expedition as reported in \cite{MilnerGulland1992} (assumed to occur at an elephant abundance of $x_0$), yielding $q=1.2\times10^{-5}$. Consider two scenarios for how price changes with respect to population abundance (1) with price highly sensitive to abundance, $b=4.7\times10^{11}$ and $z=1.5$ and (2) with price less sensitive to abundance, $b=5.3\times10^{5}$ and $z=0.5$. Both scenarios yield a current price paid to poaching gangs of $3,000$ USD per harvested elephant, as used in \cite{Lopes2015a}. We arbitrarily set $\alpha=0.0001$.

	\begin{figure}[!htb]
		\flushleft
		\begin{subfigure}[t]{0.015\textwidth}
			\textbf{a)}
		\end{subfigure}
		\begin{subfigure}[t]{0.45\textwidth}
			\includegraphics[width=\linewidth,valign=t]{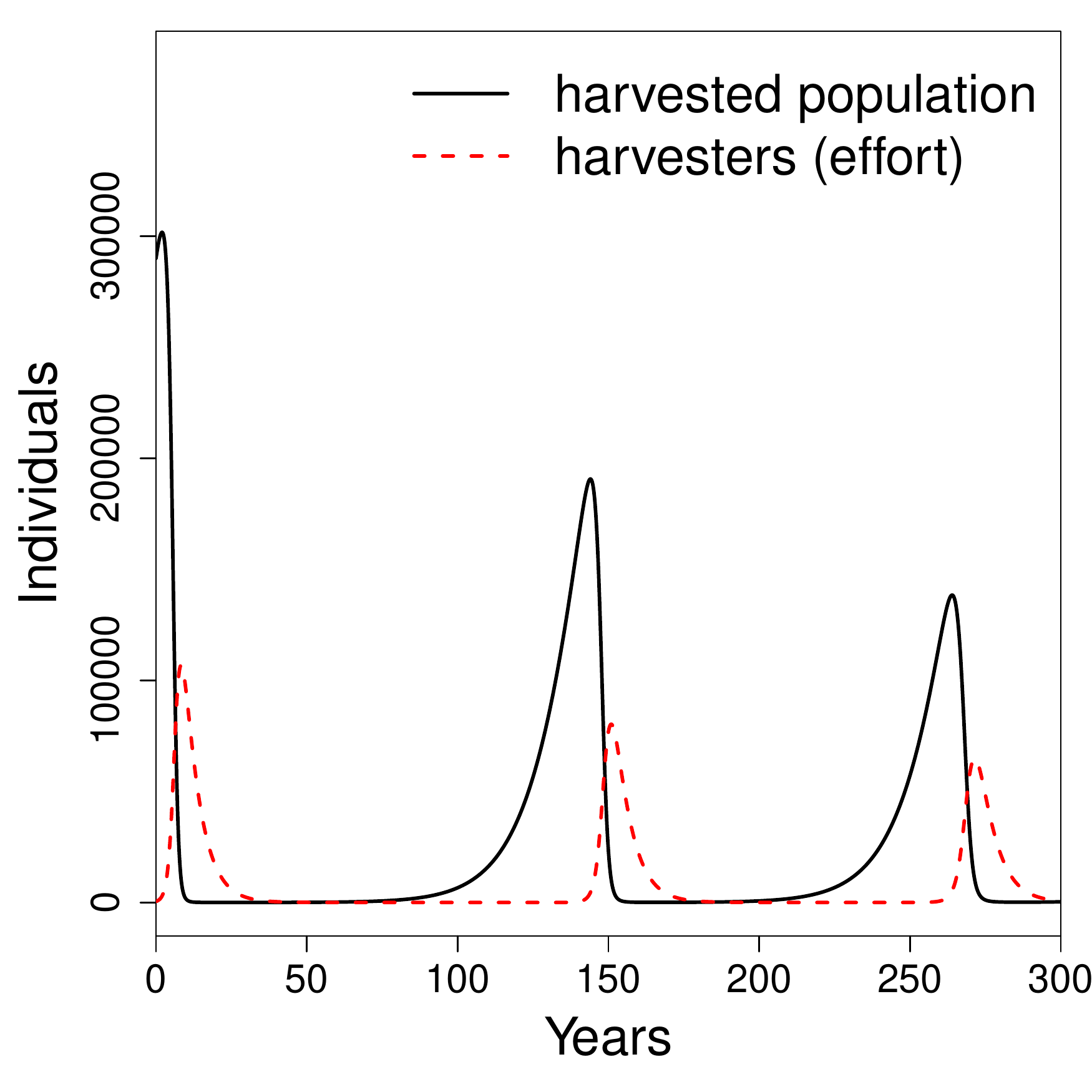}
		\end{subfigure}	
		\begin{subfigure}[t]{0.015\textwidth}
			\textbf{b)}
		\end{subfigure}
		\begin{subfigure}[t]{0.45\textwidth}
			\includegraphics[width=\linewidth,valign=t]{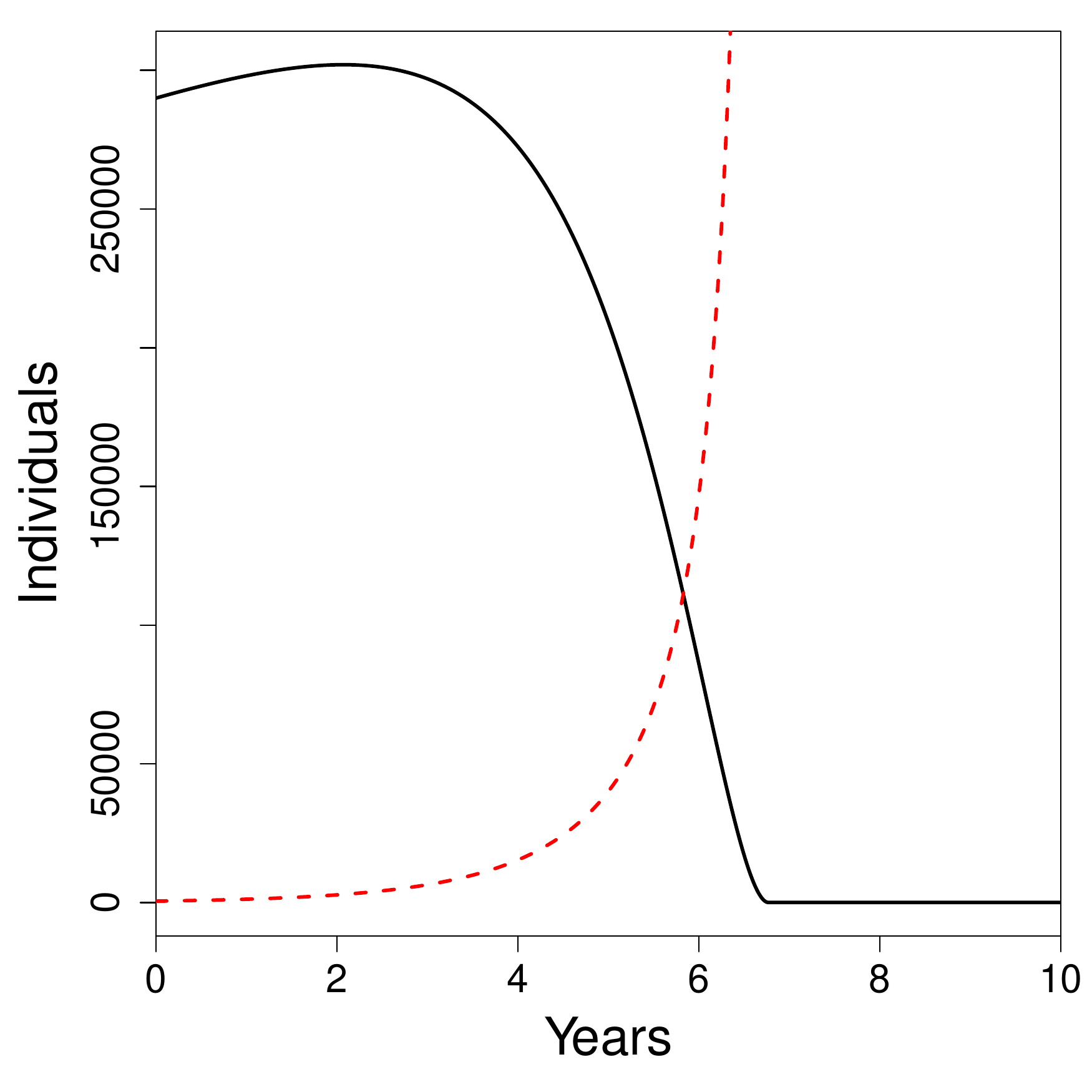}
		\end{subfigure}
		\caption{\textbf{Elephant Case Study}. Elephant population size (solid, black) and harvest effort (dashed, red) through time for (a) a less sensitive abundance price relationship, $b=5.3\times10^{5}$ and $z=0.5$ and (b) a more sensitive price abundance relationship, $b=1.5\times10^{11} $ and $z=1.5$. Other parameters are set to the baseline values for Elephants in southern Africa range states.}\label{ElephantFig}
	\end{figure}
	
	\tab In the scenario where price is less sensitive to abundance, elephant populations oscillated towards a stable equilibrium of 10,300 elephants (or about 0.03 of carrying capacity) alongside intense pulses of poaching (Fig. \ref{ElephantFig}a). When the price is more sensitive to abundance, elephant populations quickly go extinct. In this case, the positive saddle equilibrium does not exist; any potential initial elephant abundance is destined to extinction. This decline becomes increasingly more rapid in time leading to a concave\footnote{we use the word ``concave," as in a concave function to avoid confusion. However, the authors in \citep{diFonzo2013} refer to this type of decline as convex.} rather than exponential decline (Fig. \ref{ElephantFig}b), which have frequently been observed in time series data of species? abundances \citep[see][]{diFonzo2013,diFonzo2016}. This elephant example is demonstrative, as we do not have adequate data to estimate a price abundance relationship for elephant ivory. However, the example shows how important it is to learn this relationship, as small differences can mean the difference between rapid extinction and long-term persistence.
	
	\section{Model limitations, assumptions, and alternative approaches}
	\tab We found that the dynamics of a harvested population are sensitive to the description of how the price of harvested individuals changes with respect to population abundance. We must first note a key few assumptions that drove these results.
	
	\tab To achieve an AAE (with true deterministic extinction), price must become infinite as population abundance approaches zero. The condition is satisfied for all $P(x)$, with  $z>0$, studied in this paper, and the requirement is necessary because, in standard open access models, the cost of harvest approaches infinity as the stock approaches zero. If price were bounded and costs were not, then an additional low harvest equilibrium would exist, and it would be stable (or a center) because cost would be greater than price for infinitesimally small population sizes. While extinctions are possible using bounded price functions, this requires modification to other model components, such as adding an ecological Allee effect \cite{VERMA2016}. 
	
	\tab As stated previously, our model is undefined for $x=0$. As $x\rightarrow 0$, $dy/dt \rightarrow \infty$, so when simulating the dynamics, it is important to stop the simulations when $y$ reaches a high threshold value. Additionally, our model assumes population size is continuous, despite the impossibility of fractional individuals. We have chosen to use this continuous approximation because this is the language used in the foundational papers describing the AAE \cite{Courchamp2006a,hall2008} and, additionally, it is easier to describe price and demand using the units of harvested individuals rather than population densities. However, a simple rescaling of the variables to density (individuals per unit area) or biomass is straightforward, and should not affect the qualitative nature of the observed dynamics.
	
	\tab The theory behind the AAE is developed using models where price, $P(x)$, is a function of population abundance \citep{Courchamp2006a,hall2008}. In the \textit{theory of supply and demand}, population abundance can be viewed as ``potential" supply (i.e. the amount of the resource available for harvest at any given time). However, the actual supply to the market, is a function of both abundance and harvest effort, in our model, $qxy$. There are a few open-access fishery models that set price as a function of supply, $P(qxy)$, \cite{Clark1990,Auger2010,Mansal2014,Ly2014} and while none of these models predict an AAE exactly as it is described in classic AAE theory \cite{Courchamp2006a}, these models do not consider the more general class of price functions studied here. There is strong evidence that price is negatively correlated with species abundance \cite{Courchamp2006a,gault2008,angulo2009,purcell2014,hinsley2015}, but there is also evidence that the price of traded wildlife products decrease with market supply \cite[e.g. captive parrot abundance in Australian bird markets][]{vall2017}. Population abundance in the wild could strongly co-vary with total harvest supplied to the market, and therefore the most appropriate model for price may be difficult to determine.
	
	\tab Unlike most of the models that consider price as a function of demand \cite{Clark1990,Auger2010,Mansal2014,Ly2014},  \citet{Burgess2017} use a nonlinear price-yield relationship, $P(Y)=b/(Y)^z$, that can produce an AAE. In their model, yield is given by $Y=qx^\beta y$, which includes a parameter, $\beta$, representing ``catch flexibility." In the special case where $\beta=z=1$, their model reduces to the system described by \citet{Ly2014}, who showed that these equations cannot create an AAE; the equations can only produce two types of dynamics, global extinction or an approach to a stable, positive equilibrium. For the more general model, \citet{Burgess2017} show that extinction can occur, if and only if $z>\beta$ (or $z=\beta$, and $r^zq<b$). Extinction in this model includes both the possibility of an AAE or global extinction. While such a model is able to produce an AAE, the complete rigorous description of all possible trajectories in this more complicated system is an open area of research.
	
	\tab Lastly, our model only considered the harvest of a single species. An alternative explanation of population extinctions, to the AAE, is called opportunistic exploitation, where harvesters catch rare species if they happen to encounter them while harvesting more common ones \cite{branch2013}. Models that incorporate multi-species harvest and price abundance relationships have yet to be explored.

	\section{Discussion}
	
	\tab In this paper we found that standard arguments used to propose the existence of the anthropogenic Allee effect (AAE) \cite{Courchamp2006a}, while logically compelling, can provide misleading intuition. For example, in cases where classic arguments would predict sustainable harvest \cite{Courchamp2006a,hall2008}, the approach to a sustainable equilibrium can be slow and oscillatory, with minimum population sizes close to zero. In such cases, population extinction is probable, due to demographic and environmental stochasticity, or processes that reduce population growth at low abundances (i.e. ecological Allee effects) \cite{berec2007}. 
	
	\tab In cases where there is a harvest-induced Allee threshold, this threshold is not solely based on population size, as previously described in the literature \cite{Courchamp2006a,hall2008}, but also depends on initial harvest effort. Despite predictions of persistence in classic AAE theory, populations above the Allee threshold can go extinct if initial harvest effort is high. In reality, harvest effort can suddenly vary due to external forces (such as economic crises), a factor that could drive populations extinct, even when AAE theory would predict stable population sizes. 
	
	\tab When price is not allowed to decline below a critical value, populations close to carrying capacity, and suffering only minimal initial harvest, are destined to extinction. In such cases, only a small range of initial population sizes and harvest efforts lead to long-term persistent populations. This is all despite the fact that classic AAE theory would predict that these large populations would remain abundant, leading to a false sense of security. 
	
	\tab In the original AAE model \cite{Courchamp2006a}, populations that start above the Anthropogenic Allee threshold are argued never to decrease. But then how do populations ever get below this proposed threshold, to begin with? Possibilities include events unrelated to price-rarity relationships, such as natural disasters and environmental stochasticity. Our mathematical representations of the graphical models proposed in the founding AAE papers \cite{Courchamp2006a,hall2008}, reveal a potential mechanism for how harvested, abundant populations can decline below the anthropogenic Allee threshold towards extinction. The intuition is as follows: adding a minimum price received per harvested individual creates an incentive to poach when the population is abundant because cost-per-unit harvest is close to zero and price is always greater than the minimum value. This means that the separatrix, dividing the extinction and persistence basins of attraction, in the model where there is no minimum price (see green solid curve in Fig. \ref{oscFig}cd), folds downward for high population sizes, when introducing a minimum price. This creates the dome shaped persistence regions and complementary extinction regions in Fig. \ref{laFig}a and Fig. \ref{nlaFig}. 
	
	\tab The reason why increasing the speed at which harvesters adjust their effort, $\alpha$, only stretches the persistence area vertically and not horizontally, is that $\alpha$ does not affect the population sizes for which it is profitable to increase poaching effort. It only affects how quickly harvest effort increases or decreases. As proven for the case where there is no minimum price, $a=0$, $\alpha$ steepens the separatrix. However, when $a>0$, because $alpha$ simply multiplies $dy/dt$, increasing $\alpha$ not only steepens the initial slope of the separatrix at the saddle equilibrium, it also increases $dy/dt$ to the right of the second equilibrium, making the downward folding of the separatrix equally as steep. 
	
	\tab The difference in population trajectories produced by models with only subtle differences in price abundance relationships is alarming. Data for both prices of wildlife products in combination with estimates of population abundance are difficult to obtain and are likely sensitive to many external social, economic and environmental factors. While the AAE is used to inform conservation decisions \cite[e.g.][]{harris2013}, we warn that without an understanding of how price is affected by species abundance, it may be difficult to predict population responses to conservation interventions and thus determine the best management actions to protect overexploited species. Therefore, it is important to consider a wide range of likely scenarios to make sure decisions are robust to differences in price-abundance relationships. If the best decisions are not robust across scenarios, one must consider the value of resolving model uncertainty \citep{maxwell2015}, in price abundance relationships, for making effective management decisions.

	\pagebreak
	\section{References}
	\bibliographystyle{jtb}	
	\bibliography{AAE_matt}

	\section{Appendix}
	\subsection{Theorems for Linear Growth Model}
	\begin{thm}\label{oscThm}
		If $0\leq z<1$, the positive equilibrium \eqref{equil} is a center, surrounded by infinitely many closed periodic orbits.
	\end{thm}
	\begin{proof}
		From \eqref{dynamics} 
		\begin{equation}
		\frac{dx}{dy}=\frac{x(r-qy)}{y(bqx^{1-z}-c)}
		\end{equation}
		By seperation of variables the solution to the above is the solution, $(x,y)$, to  
		\begin{equation}
		x^{-c}e^{\frac{bqx^{1-z}}{1-z}} = \mathscr{C}y^re^{-qy}
		\end{equation} 
		where $\mathscr{C}$ is a constant. Equating the left and right hand side shows that the solutions are closed orbits. That is let 
		\begin{align}
		w_1(x) &= x^{-c}e^{\frac{bqx^{1-z}}{1-z}} \\
		w_2(y) &= \mathscr{C}y^re^{-qy}
		\end{align}
		Note that if $0\leq z<1$, $w_1(x)\rightarrow\infty$ as both $x\rightarrow0$ and as $x\rightarrow\infty$, and that $dw_1/dx<0$ for $x<\left(\frac{c}{b q}\right)^{\frac{1}{1-z}}$ and $dw_1/dx>0$ for $x>\left(\frac{c}{b q}\right)^{\frac{1}{1-z}}$. Similarly, $w_2(y)\rightarrow0$ as both $y\rightarrow0$ and as $y\rightarrow\infty$, and $dw_2/dy<0$ for $y>r/q$ and $dw_2/dy>0$ for $y<r/q$. Therefore the solution $(x,y)$ to $w_1(x)=w_2(y)$ is a closed orbit.
	\end{proof}
	
	\begin{prop} \label{saddleThm}
		If $z>1$, the positive equilibrium \eqref{equil} is a saddle.
	\end{prop}
	\begin{proof}
		The determinant of the Jacobian matrix for the system \eqref{dynamics}, evaluated at \eqref{equil}, is always negative for $z>1$.
	\end{proof}
	
	The slope of the ``separatrix" near the saddle equilibrium, written in \eqref{slope}, is solved by calculating the slope of the eigenvector (corresponding to the negative eigenvalue) of the Jacobian matrix for the system \eqref{dynamics}, evaluated at the saddle equilibrium \eqref{equil}. Which we computed using Mathematica. We note the following theorem about this slope, which states that the separatrix is steepest for high values of $\alpha, r$ and $c$ and more shallow for high $b$ and $q$.
	
	\begin{prop} \label{eigvecThm}
		If $z>1$, the eigenvector of the Jacobian corresponding to the stable manifold of the saddle equilibrium \eqref{equil} decreases in steepness with respect to $b$ and $q$, and increases in steepness with respect to $\alpha,r$, and $c$ and also for $z$ if $c<bq$.
	\end{prop}
	\begin{proof}
		With a bit of algebra the slope of the eigenvector can be rewritten as 
		\begin{equation} \label{eigvec}
		b^\frac{1}{1-z}c^\frac{z+1}{2z-2}q^\frac{z}{1-z}\sqrt{\alpha r (z-1) }
		\end{equation}
		which is clearly increasing with $\alpha, r$, and $c$ and decreasing with $b$ and $q$. The derivative of \eqref{eigvec} with respect to $z$ is
		\begin{equation}
		\frac{\alpha b r \left(\frac{c}{b q}\right)^{\frac{z}{z-1}} \sqrt{  z-1-\log \left(\frac{c}{b q}\right)}}{c (z-1)},
		\end{equation}
		which is always positive if $z>1$ and $c<bq$. Hence, the slope of the eigenvector increases with $z$, if $z>1$ and $c<bq$.
	\end{proof}

	\subsection{Stability analysis for logistic growth model}
	\begin{prop}
		If $c > bqk^{1-z}$, the equilibrium $x=k,y=0$ is locally stable and if $c < bqk^{1-z}$, the equilibrium $x=k,y=0$ is a saddle
	\end{prop}
	\begin{proof}
		The trace of the Jacobian evaluated at $x=k,y=0$, is $b q k^{1-z}-c-r$ and the determinant is $r \left(c-b q k^{1-z}\right)$. Therefore the determinant is negative for $c < bqk^{1-z}$ meaning $x=k,y=0$ is a saddle. The determinant is positive and trace is negative for $c > bqk^{1-z}$ meaning $x=k,y=0$ is stable.
	\end{proof}

	\begin{prop} \label{saddleThmNL}
		If $z>1$, and $c<bqk^{1-z}$, the positive equilibrium \eqref{nlequil} is a saddle, with $x^*<k$. 
	\end{prop}
	\begin{proof}
		The determinant of the Jacobian evaluated at \eqref{nlequil} has the same sign as the quantity
		\begin{equation}
		D \equiv c\left(\frac{c}{bq}\right)^\frac{z}{z-1}+bq\left([1-z]k+[z-2]\left(\frac{c}{bq}\right)^\frac{z}{z-1}\right).
		\end{equation}
		For $z>0, z\neq 1$, $D=0$ $\iff$ $k = (bq/c)^{\frac{1}{z-1}}=x^*$. If $z>1$, $D$ is clearly negative as $k\rightarrow\infty$ and therefore by continuity, $D<0$ if $k>(bq/c)^{\frac{1}{z-1}}$, which can be equivalently written as $c<bqk^{1-z}$. This proves \eqref{nlequil} is a saddle.
	\end{proof}
	
	Note that if $z>1$ and $c>bqk^{1-z}$ this equilibrium has $x^*>k$ and $y*<0$ and is therefore not biologically relevant. The fact that $(k,0)$ is unstable and the only positive equilibria, means that the population goes extinct and harvest effort approaches infinity.
	
	\begin{prop} \label{stableThmNL}
		If $0<z<1$, and $c>bqk^{1-z}$, the positive equilibrium \eqref{nlequil} is locally stable, with $x^*<k$.
	\end{prop}
	\begin{proof}
		By a similar argument to \ref{saddleThmNL}, if $0<z<1$ then $D>0$ and $x^*<k$. For $z\neq1$, the trace of the Jacobian evaluated at \eqref{nlequil} is $T=-r(bq/c)^{\frac{1}{z-1}}/k$, which negative for all parameter values. $D>0$ and $T<0$ means the equilibrium is stable.
	\end{proof}
	
	\subsection{Stability analysis for linear growth model with minimal price, $a$}
	\begin{prop} \label{propLaequil1}
		If $c > 2q\sqrt{ab}$, and $z=2$, the equilibrium \eqref{laequil1} is a saddle
	\end{prop}
	\begin{proof}
		The trace of the Jacobian evaluated at the equilibrium \eqref{laequil1} has a determinant of
		\begin{equation}
		\frac{\alpha q \left(c \sqrt{c^2-4abq^2}+4abq^2-c^2\right)}{\sqrt{c^2-4 abq^2}-c}
		\end{equation}
		which is negative for all $c > 2q\sqrt{ab}$.
	\end{proof}

	\subsection{Stability analysis for logistic growth model with minimal price}
	
	\begin{prop} \label{KThmNL}
		If $c > aqk+bqk^{1-z}$, the equilibrium $x=k,y=0$ is locally stable and if $c < aqk+bqk^{1-z}$, the equilibrium $x=k,y=0$ is a saddle
	\end{prop}
	\begin{proof}
		The trace of the Jacobian evaluated at $x=k,y=0$, is $aqk + b q k^{1-z}-c-r$ and the determinant is $r \left(c - aqk - b q k^{1-z}\right)$. Therefore the determinant is negative for $c < aqk+bqk^{1-z}$ meaning $x=k,y=0$ is a saddle. The determinant is positive and trace is negative for $c > aqk+bqk^{1-z}$ meaning $x=k,y=0$ is stable.
	\end{proof}
	
	\pagebreak
	\section{Author Contributions:} MHH designed and carried out the research, and wrote the paper. EM-M discussed the results with MHH and contributed to manuscript revisions.
	
	\section{Acknowledgments:} MHH is funded by a postdoctoral fellowship from the Australian Research Council Centre of Excellence for Environmental Decisions. We thank two anonymous reviewers for their valuable feedback.
	\pagebreak
	
	\beginsupplement
	\section{Online supplemental figures}
	
	\begin{figure}[!htbp]
		\flushleft
		\begin{subfigure}[t]{0.015\textwidth}
			\textbf{a)}
		\end{subfigure}
		\begin{subfigure}[t]{0.3\textwidth}
			\includegraphics[width=\linewidth,valign=t]{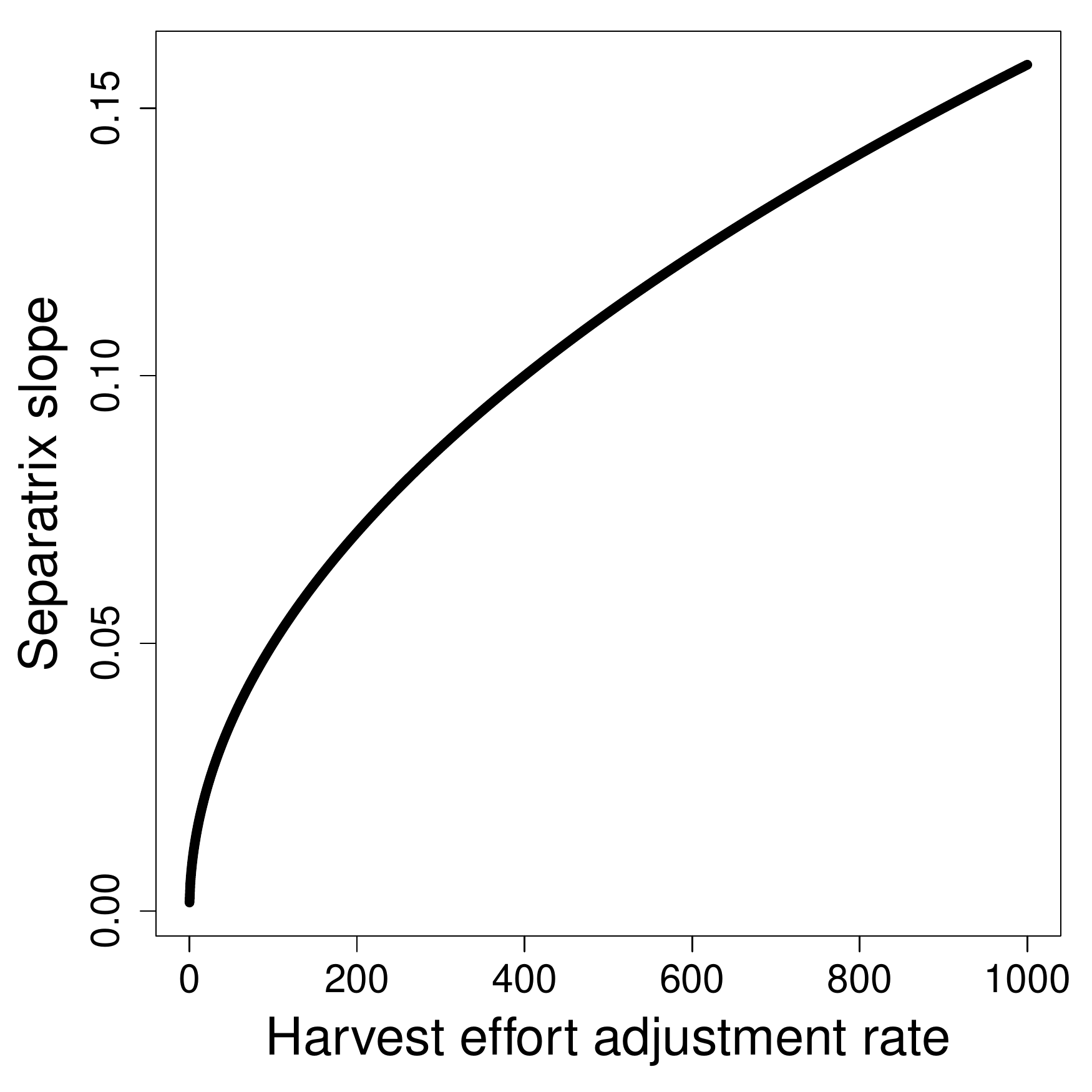}
		\end{subfigure}	
		\begin{subfigure}[t]{0.015\textwidth}
			\textbf{b)}
		\end{subfigure}
		\begin{subfigure}[t]{0.3\textwidth}
			\includegraphics[width=\linewidth,valign=t]{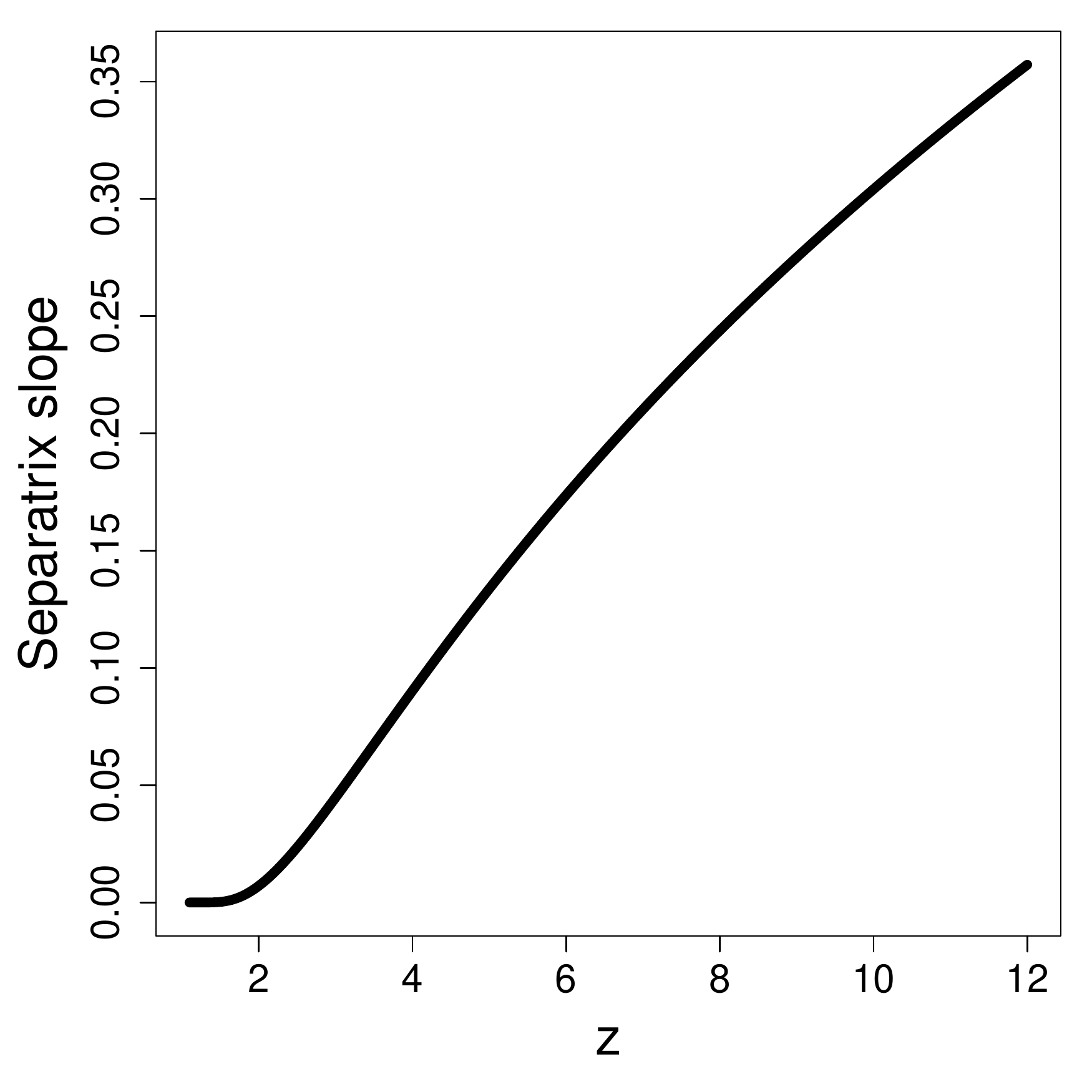}
		\end{subfigure}	
		\begin{subfigure}[t]{0.015\textwidth}
			\textbf{c)}
		\end{subfigure}
		\begin{subfigure}[t]{0.3\textwidth}
			\includegraphics[width=\linewidth,valign=t]{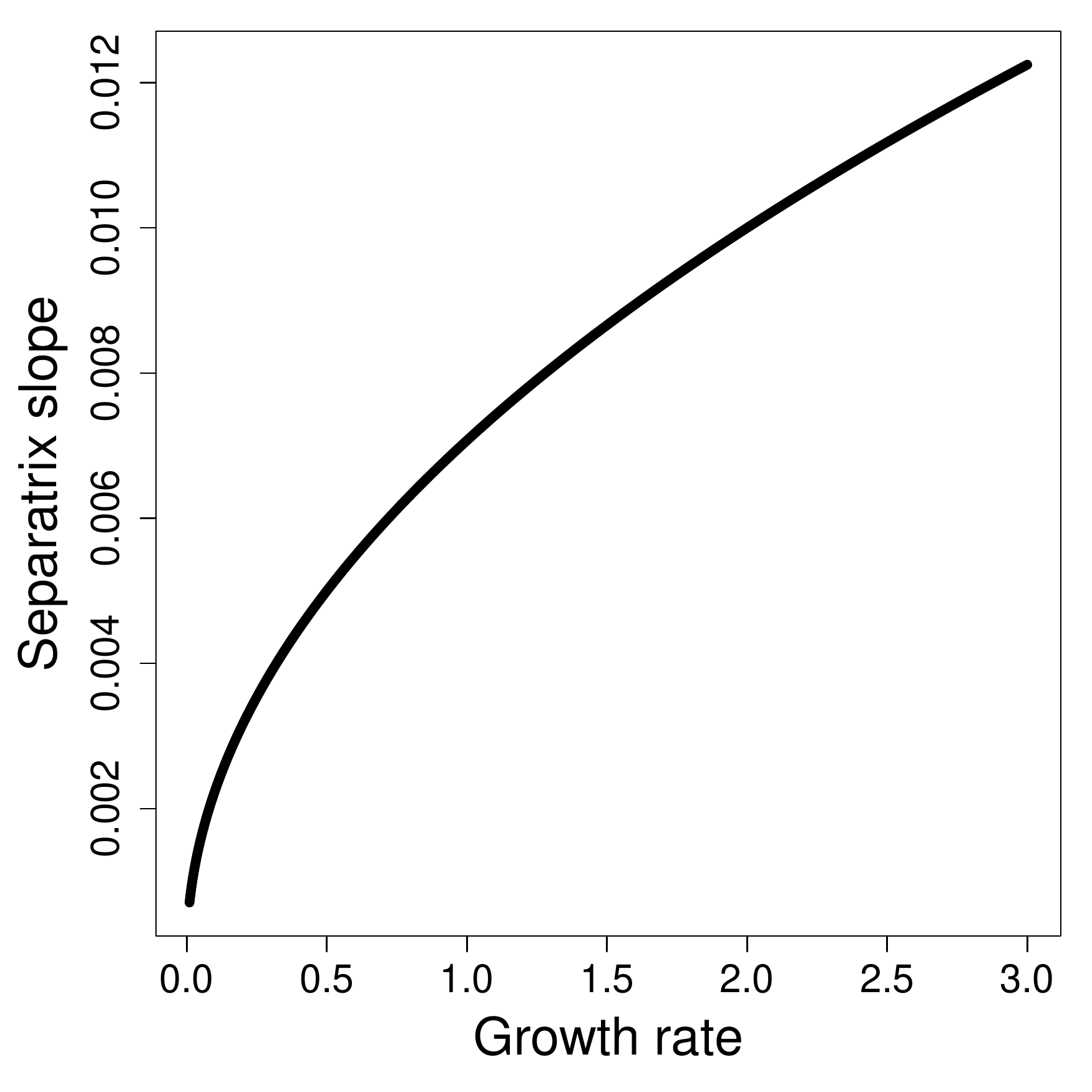}
		\end{subfigure}\\
		\begin{subfigure}[t]{0.015\textwidth}
			\textbf{d)}
		\end{subfigure}
		\begin{subfigure}[t]{0.3\textwidth}
			\includegraphics[width=\linewidth,valign=t]{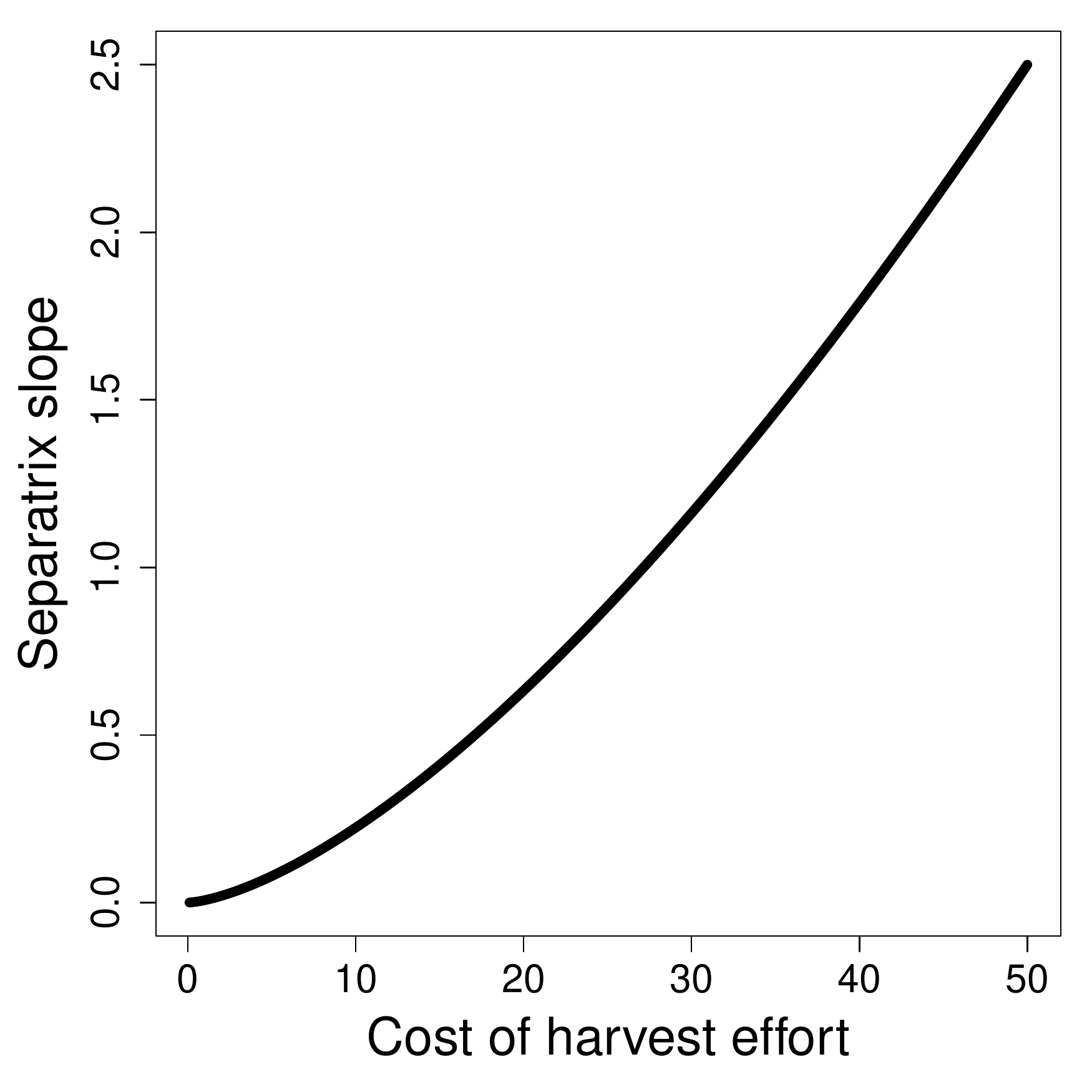}
		\end{subfigure}
		\begin{subfigure}[t]{0.015\textwidth}
			\textbf{c)}
		\end{subfigure}
		\begin{subfigure}[t]{0.3\textwidth}
			\includegraphics[width=\linewidth,valign=t]{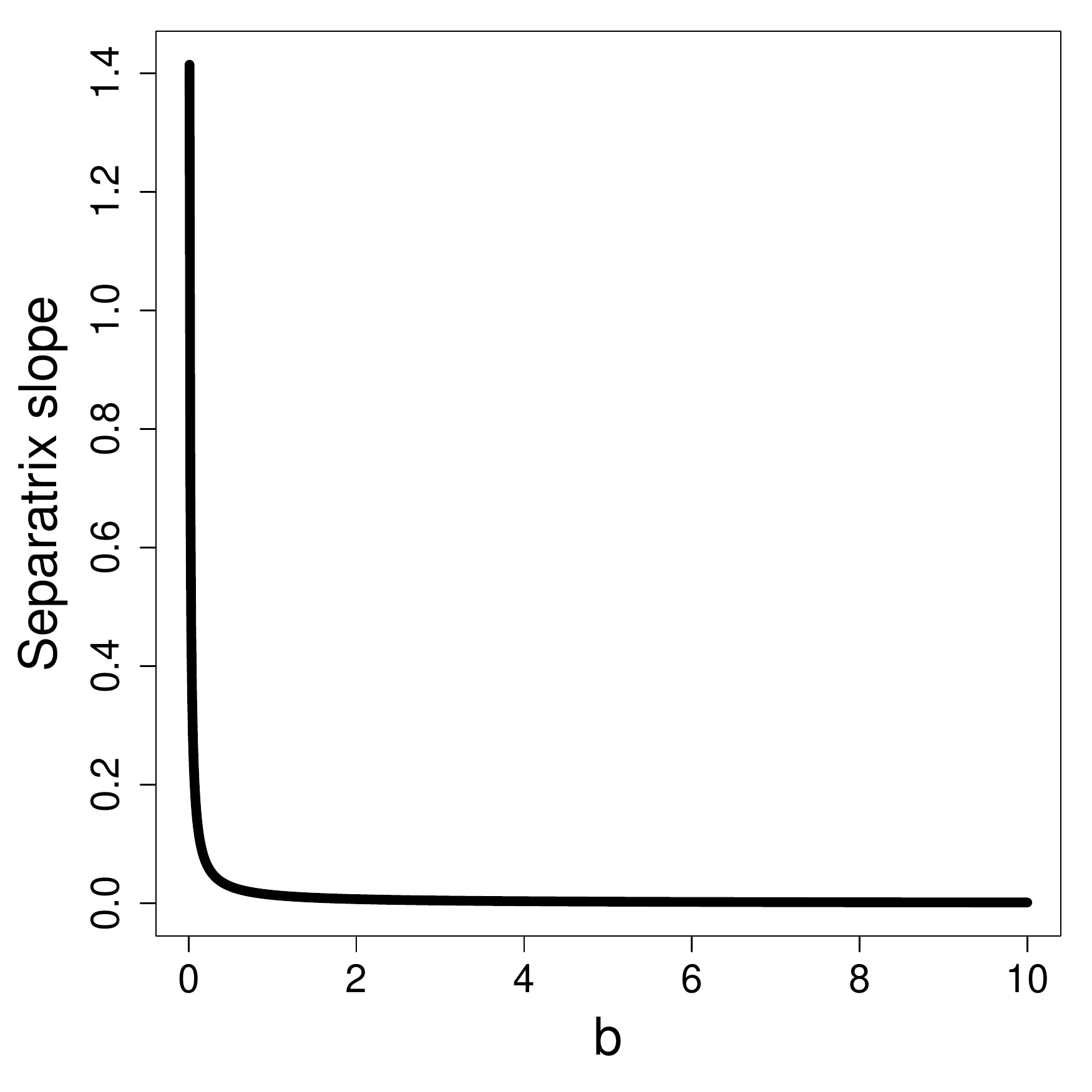}
		\end{subfigure}
		\begin{subfigure}[t]{0.015\textwidth}
			\textbf{d)}
		\end{subfigure}
		\begin{subfigure}[t]{0.3\textwidth}
			\includegraphics[width=\linewidth,valign=t]{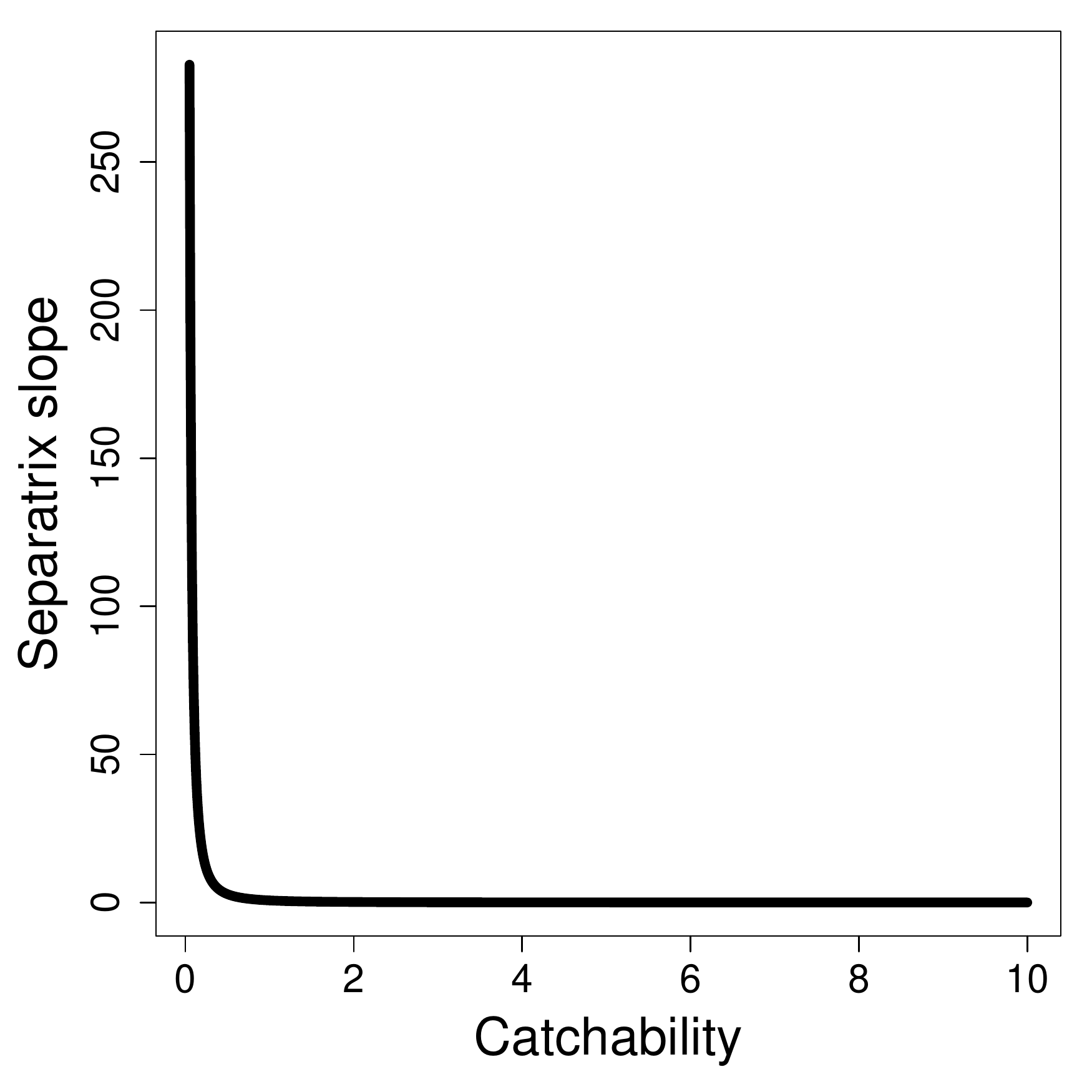}
		\end{subfigure}
		\caption{\textbf{Steepness of the separatrix} vs. different values of (a) $\alpha$, (b) $z$, (c) $r$, (d) $c$, (e) $b$, (f) $q$.  With other parameters fixed at $\alpha=2$, $z=2$,$r = 1, c = 1,	q = 10$ and $b = 2$. Note that slopes above 0.1 are quite steep (because harvest effort is two orders of magnitude less than population size near the AAE equilibrium).}\label{steepParms}
	\end{figure}
	
	\begin{figure}[h]
		\includegraphics[width=.46\textwidth]{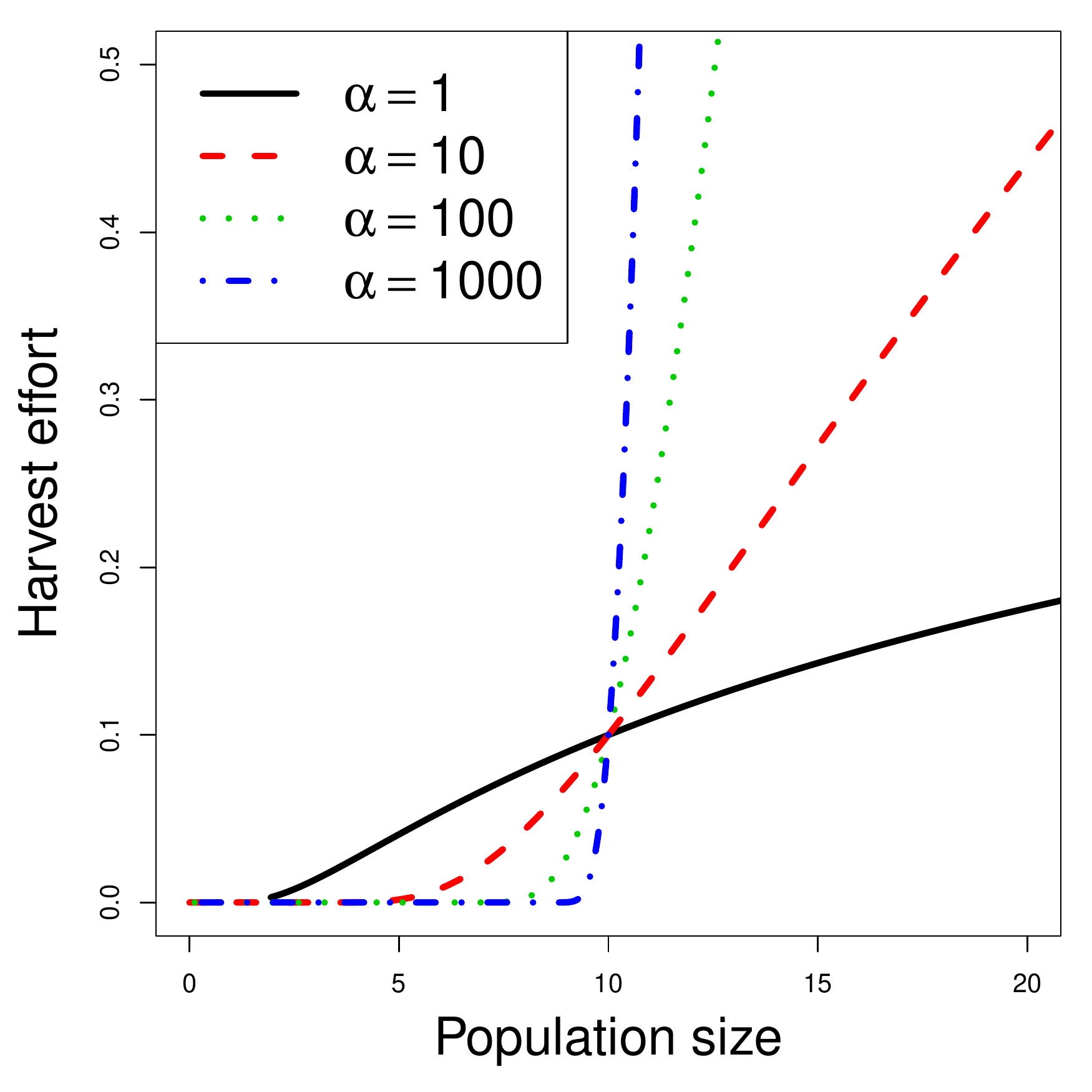}
		\caption{The separatrix (stable manifold of the saddle), for different values of $\alpha$. If harvesters adjust effort slowly (low $\alpha$, black solid line) the curve is flat (low steepness), meaning initial harvest effort determines whether populations go extinct, and if harvesters adjust effort quickly (high $\alpha$, blue dot-dashed line), classic AAE theory is realized, and final population size depends solely on initial population size (high steepness). Other parameters are $r=1, q=10, c=1, b=1$.}\label{steep}
	\end{figure}

\end{document}